\documentclass[acmsmall, screen]{acmart}
\AtBeginDocument{%
  }

\setcopyright{acmlicensed}
\acmDOI{}
\acmISBN{978-1-4503-XXXX-X/2018/06}

\usepackage{listings}
\usepackage{xcolor}
\usepackage[ruled,vlined,linesnumbered,commentsnumbered]{algorithm2e}
\usepackage{amsthm}
\usepackage{amsmath}
\usepackage{multirow}
\usepackage[]{hyperref}
\usepackage{multirow}
\usepackage{wrapfig}
\usepackage[most]{tcolorbox}

\theoremstyle{definition}
\newtheorem{definition}{\textbf{Definition}}[section]

\theoremstyle{remark}

\definecolor{codegreen}{rgb}{0,0.6,0}
\definecolor{codegray}{rgb}{0.5,0.5,0.5}
\definecolor{codepurple}{rgb}{0.58,0,0.82}
\definecolor{backcolour}{rgb}{0.95,0.95,0.92}
\definecolor{codeblue}{rgb}{0,0,1}
\definecolor{codegold}{rgb}{0.85,0.65,0.13}

\lstdefinestyle{mystyle}{
    backgroundcolor=\color{backcolour},   
    commentstyle=\color{codegreen},
    keywordstyle=\color{magenta},
    numberstyle=\tiny\color{codegray},
    stringstyle=\color{codepurple},
    basicstyle=\ttfamily\footnotesize,
    alsoletter={_},                 
    emphstyle=[2]\color{codeblue},   % functions in blue
    emphstyle=[3]\color{codegold},   % inner calls in golden
    emphstyle=[4]\color{magenta},
    breakatwhitespace=false,         
    breaklines=true,                 
    captionpos=b,                    
    keepspaces=true,                 
    numbers=left,                    
    numbersep=5pt,                  
    showspaces=false,                
    showstringspaces=false,
    showtabs=false,                  
    tabsize=2
}

\begin{document}

%%
%% The "title" command has an optional parameter,
%% allowing the author to define a "short title" to be used in page headers.
\title[Phaedrus: Predicting Dynamic Application Behavior with Lightweight Generative Models and LLMs]{Phaedrus: Predicting Dynamic Application Behavior with Lightweight Generative Models and LLMs}

%%
%% The "author" command and its associated commands are used to define
%% the authors and their affiliations.
%% Of note is the shared affiliation of the first two authors, and the
%% "authornote" and "authornotemark" commands
%% used to denote shared contribution to the research.
\author{Bodhisatwa Chatterjee}
\orcid{0000-0003-3098-6256}
\email{bodhi@gatech.edu}
\affiliation{%
  \institution{Georgia Institute of Technology}
  \city{Atlanta}
  \country{USA}
}
\author{Neeraj Jadhav}
\orcid{0000-0003-3098-6256}
\email{njadhav35@gatech.edu}
\affiliation{%
  \institution{Georgia Institute of Technology}
  \city{Atlanta}
  \country{USA}
}

\author{Santosh Pande}
\orcid{0000-0001-6723-8062}
\email{santosh@cc.gatech.edu}
\affiliation{%
  \institution{Georgia Institute of Technology}
  \city{Atlanta}
  \country{USA}
}

%%
%% By default, the full list of authors will be used in the page
%% headers. Often, this list is too long, and will overlap
%% other information printed in the page headers. This command allows
%% the author to define a more concise list
%% of authors' names for this purpose.
\renewcommand{\shortauthors}{Chatterjee et al.}

%%
%% The abstract is a short summary of the work to be presented in the
%% article.
\begin{abstract}
Application profiling is an indispensable technique for many software development tasks, such as code and memory layout optimizations, where optimization decisions are tailored to 
specific program profiles. Unfortunately, modern application codebases exhibit highly variant behavior across different inputs, creating challenges for conventional profiling approaches that rely on a single representative execution instance. In this paper, we propose \textbf{Phaedrus}, a new \textit{compiler-assisted deep learning framework} designed to predict dynamic program behaviors across varied execution instances, specifically focusing on dynamic function call prediction. These predicted call sequences are subsequently used to guide input-specific compiler optimizations, producing code specialized for each execution instance.

Traditional profile-guided optimization methods struggle with the input-dependent variability of modern applications, where profiling on different inputs yields divergent application behaviors. To address this, Phaedrus proposes two new approaches: \textit{Application Behavior Synthesis (Dynamis)}, a profile-less approach where Large Language Models (LLMs) directly infer dynamic functions based on source code \& static compiler analysis, bypassing the need for traditional profiling, and \textit{Application Profile Generalization (Morpheus)}, which uses generative models trained on compressed and augmented \textit{Whole Program Path} (WPP) based function profiles to predict application behavior under unseen inputs. Our experiments show that \textit{Phaedrus} accurately identifies the most frequently executed and runtime-dominated hotspot functions, accounting for up to 85–99\% of total execution time. Leveraging these predictions, \textit{Phaedrus} enables superior profile-guided optimizations, delivering an average speedup of 6\% (upto 25\%) and a binary size reduction of 5.19\% (upto 19\%), without any program execution. In addition, \textit{Phaedrus} reduces WPP function profile sizes by up to $10^{7}\times$.

\end{abstract}

\received{20 February 2007}
\received[revised]{12 March 2009}
\received[accepted]{5 June 2009}

%%
%% This command processes the author and affiliation and title
%% information and builds the first part of the formatted document.
 \maketitle

\section{Introduction}
\label{sec:intro}

Application profiling and execution tracing are crucial in the software development process. Popular performance enhancement techniques such as \textit{profile-guided optimization} (PGO) \cite{radigan1995integer,  grove1995profile, gupta1997resource} involve generating execution traces for applications and then performing optimization decisions specific to the collected profile, such as code transformation \cite{chang1991using, visochan2022method, sydow2020towards, joshi2004targeted}, data layout restructuring and reorganization \cite{pettis1990profile}. Application profiles and traces are also used to accomplish a wide variety of tasks such as estimating resource usage for \textit{scheduling and resource allocation} \cite{agelastos2015toward, song2009energy, viswanathan2011energy, rashti2015wattprof, agelastos2016continuous, kestor2013enabling}, \textit{memory management for scratchpads, caches, and other memory hierarchies} \cite{cho2009adaptive, ansari2009profiling, bruno2017polm2, agarwal2017thermostat, xu2024prompt}, \textit{memory simulation} \cite{uhlig1997trace, uhlig2001trace, rico2011trace, hassan2007synthetic} and \textit{debugging} \cite{arnold2007stack, nagarajan2012system,ning2017ninja,ning2018hardware,verma2020interactive}. Such profile-driven approaches have been successfully utilized and deployed for over two decades. 

Most profile-driven tasks rely on the users to manually select profiles representative of the application's intended execution scenario. This ensures the usefulness of the profile analysis and optimization decisions during the application's runtime across different inputs. However, recent work \cite{chatterjee2022cas, mururu2023beacons, ahmed2017leveraging} has shown that most modern applications are complex and exhibit a high amount of input dependence, resulting in {\it variable} application behavior across different execution instances on various input data sets. This variant behavior is exhibited in all program components, including program branches, loops, and function calls. As a result, the application profile collected during a particular execution instance on a given input data set can be entirely different from another execution instance collected using a different input data set. A naive approach to tackle this problem would be to profile the application on multiple inputs - in practice, however, this `excessive' profiling approach does not scale for workloads whose input domain is not numerable, such as graph applications. In addition, profile mixing can also be problematic for specific path-based optimizations that optimize along hot trace paths, pushing overheads to off-trace paths.

\textit{In this work, we study the problem of generating dynamic application behavior across different execution scenarios pertinent to different input data sets by capturing traces of application's function calls and control flow behavior}. Specifically, we explore the problem of \textbf{hot function call prediction}, where the objective is to obtain the set of computationally intensive (hot) functions called by the application during its runtime, specific to a given program input. In this paper, we propose \textit{Phaedrus}, a new compiler-assisted deep learning framework which employs two different approaches to predict this dynamic program behavior: (1) \textit{application behavior synthesis}, where a compiler-assisted LLM framework reasons about the applications' algorithmic properties and extrapolates the application's dynamic set of \textit{hot} functions without any profiles; (2) \textit{application profile generalization}, where lightweight generative models are synthesized to predict the application's dynamic function call chains by learning a specialized representation of the smallest application profile. \textit{Phaedrus}'s end goal is to use this predicted dynamic behavior to perform profile-guided optimization (PGO), such as changing the application function layout based on most executed functions to seek code size and performance enhancements.

\textbf{Phaedrus Framework:} \textit{In a nutshell, the idea of this work is to utilize generative models to predict the behavior of dynamic applications}. The predictions from these models are then further used to perform profile-guided optimization. Towards this, the first new concept introduced in this work is \textit{application behavior synthesis}, which attempts to reason about the dynamic program behavior without execution or collected profiles with the aid of Large-Language Models (LLMs). Owing to their remarkable successes in every domain, from natural language processing to software engineering, LLMs have been extensively used for generating and optimizing programs \cite{ugare2024improving, gu2023llm, lin2024llm, liu2024exploring}. However, the primary challenge that prevents the adoption of LLM in place of the mainstream compilation pipeline, is the absence of correctness guarantees \cite{zhong2024chatgptreplacestackoverflowstudy, spiess2024calibrationcorrectnesslanguagemodels}. Secondly, LLMs, being generative models at their core, cannot inherently perform static program analysis, which can enhance its decision-making capacities. Thus, naively querying the LLM  for most code reasoning tasks fails spectacularly (as shown in fig. \ref{fig:6}). 

\begin{figure*}[htbp]
\centering\includegraphics[width=0.9\textwidth]{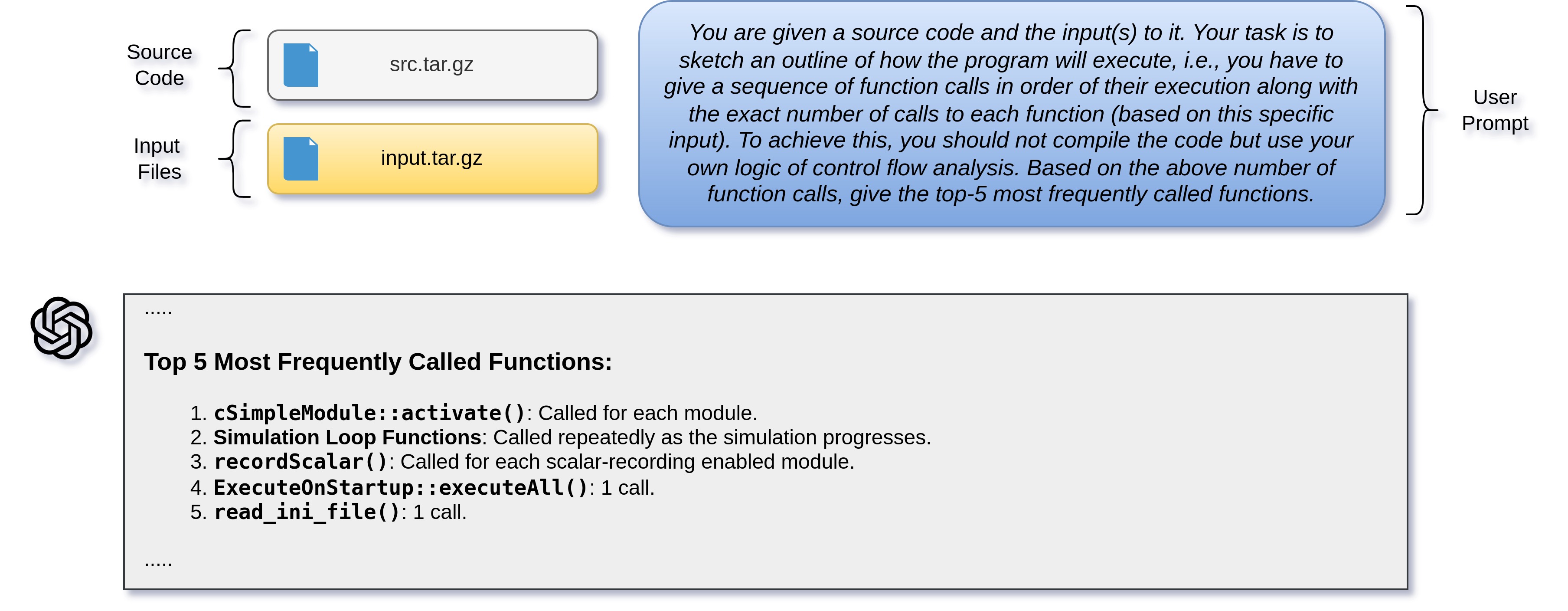}
\caption{An example of unsuccessful dynamic program analysis with the LLM, given the source code and input files (\textit{520.omnetpp\_r}). The LLM is unable to pinpoint \textit{hot} functions in the benchmark just by observing the source code \& input files, and the output is quite unsatisfactory.}
\label{fig:6}
\end{figure*}

\textbf{Compiler-Assisted LLM Reasoning}: \textit{The first contribution of this work is the utilization of Large-Language Models, supplemented by the knowledge of static program properties obtained by compiler analysis}. The \textit{Dynamis} framework operates in three broad phases: first, filtering the wide corpus of LLM knowledge through \textit{domain knowledge deduction} (\S \ref{sec:domain}); second, supplementing this understanding with \textit{program structural artifacts} such as the universe of functions and call graphs (\S \ref{sec:semantic}); and finally, fortifying the LLM’s reasoning with \textit{static compiler analysis} (\S \ref{sec:staticLLM}). While recent efforts have focused on using LLMs to generate optimized code \cite{gao2024searchbasedllmscodeoptimization, qiu2024efficientllmgeneratedcoderigorous, rosas2024aioptimizecodecomparative}, such code often requires extensive verification \cite{liu2023, tang2024, ngassom2024}. Unlike these approaches, \textit{Dynamis}  leverages LLMs only to identify input-specific profile-driven optimization opportunities, followed by traditional code generation techniques to elicit performance.

However, a downside of using LLMs for determining input-specific program behavior, is resource requirements needed for successful deployment and inference. For instance, deploying a state-of-art LLM such as \textit{LLaMA-2-70B} with 70 billion parameters for inference, can take multiple NVIDIA A100 GPUs \cite{zhou2024survey}. Secondly, LLMs often have a fixed context window that determines the amount of tokens that it can consider to answer a certain query (8192 tokens for GPT-4). Large applications with hundreds of KLOC can easily exceed this limit, and lead to undesirable results. Most importantly, with the set of prompts used in the current work, we observe that LLMs are unable to reason about the dynamic sequence of function calls satisfactorily (\S \ref{res:vs}). In the absence of accurate function call sequence, most path-based optimizations and execution forensics cannot be achieved.   

To tackle these issues, we can use application profiles that capture the entire dynamic program behavior. This leads to the second important contribution  of this work, \textit{application profile generalization}, where the primary idea is to profile the application on a convenient single (small) input, transform it by adding program context (such as knowledge of possible program paths), and then train generative models on the compressed profile, which can then predict the application behavior on an unseen execution scenario (new input).

As a lightweight alternative to the LLM-based approach, in this work, we leverage \textit{Whole Program Paths} (WPP) based function profile to make decisions based on the program's runtime execution behavior. \textit{Whole Program Path Profiles} \cite{larus1999whole, ball2000using, zhang2001timestamped,  tallam2005extended} captures the entire dynamic control flow of an application during its execution, and are superior in providing insight about the dynamic execution behavior than popular sampling-based trace collection methods \cite{martonosi1993effectiveness, kessler1994comparison, novillo2014samplepgo,liu2016sample}. However, the biggest challenge that prevents the adoption of \textit{WPP Profiles} for profile-driven tasks in modern workloads is its {\it intractability} with respect to data size and volume (\S \ref{res:morph}). Although compression of WPP profiles has been explored extensively in compiler literature \cite{larus1999whole, zhang2001timestamped, burtscher2003compressing, zhang2002path, law2003whole, milenkovic2003stream, zhang2004whole, lin2007design}, unfortunately such compression schemes can only be applied to a given WPP profile obtained from a single input data set at a time and cannot be used to get a generalized representation of profiles across different inputs (\S \ref{sec:back1}).

\textbf{Lightweight Prediction of Dynamic Program Behavior} (\S \ref{sec:loop} - \S\ref{sec:aug}): \textit{An important contribution of this work is a new profile compression scheme, which is lossless and is consistent across different inputs}. A \textit{WPP based function profile} of the application is first generated on a convenient (smallest) input. While collecting the profile, function calls present inside loops are handled with \textit{smart instrumentation} (\S \ref{sec:loop}) to avoid a high degree of redundant tokens. Compressible regions in the generated profile are identified by performing novel compiler analysis based on control-dependence and post-dominance relations (\S \ref{sec:static}). These analyses are performed statically at the IR level to obtain compressible regions that are consistent across different inputs. The compressed trace is then \textit{augmented} with additional possible control-flow paths and program context initially absent in the profile. This leads to a generalized application profile, which summarizes the program behavior across different execution settings. 

The unified \textit{WPP based function profile} serves as training data for this simple generative model, which, in turn, can predict program behavior specific to different inputs and execution phases. In this work, we show that these diverse execution behaviors can be learned with high prediction accuracy, by designing lightweight models that do not require any specialized hardware accelerators (GPU, TPU, etc.) for training / inference (\S \ref{res:morph}). We also demonstrate the utility of predicting the dynamic function sequence (and the function set) called during execution in determining hot functions for optimizations (\S \ref{res:morph}). The profile processing (generation, compression and augmentation), and the compiler analysis and model training, has been implemented as an end-to-end compiler-assisted deep-learning framework, \textit{Morpheus} (\S\ref{sec:comp}).

The experimental evaluation (\S \ref{sec:eval}) of \textit{Phaedrus} is done on the SPEC CPU 2017 benchmark suite \cite{bucek2018spec}. Our evaluation focuses on \textit{Dynamis} and \textit{application behavior synthesis}, we evaluate the prediction of hot functions and the coverage of several function calls (\S \ref{res:dyn}). We also show that \textit{Phaedrus} can achieve application binary size reduction of upto \textbf{65\%} (13.68\% on average) and performance improvement of upto \textbf{13.34\%} (2.8\% on average). For \textit{Morpheus} and \textit{application profile generalization}, our evaluation focuses on the order of compression for WPP based function profiles (upto $10^7 \times$ reduction), prediction of hot functions and its execution coverage and training overheads (\S \ref{res:morph}). To the best of our knowledge, this is the first work to explore dynamic application behavior across different execution scenarios and leverages compiler-assisted lightweight deep learning models and LLMs for predicting new program behaviors for profile-guided tasks.
.
\section{Background and Motivation}
\label{sec:back}

In this section, we first describe the describe the rationale behind leveraging LLMs for predicting dynamic application behavior. We then motivate the need for a lightweight profile-driven approach, and explain the necessity of utilizing \textit{whole program path} (WPP) profiles, and generative models for predicting dynamic application behavior. We also elaborate on the need of an \textit{input consistent} compression scheme, and \textit{incorporating program context} in \textit{WPP} profiles.

\subsection{Can LLMs reason about the dynamic program behavior?}

Large Language Models (LLMs) exhibit remarkable capabilities in terms of code understanding \cite{fang2024largelanguagemodelscode, nam2024usingllmhelpcode, wang2023codet5opencodelarge} and analysis \cite{fang2024largelanguagemodelscode, ma2024lmsunderstandingcodesyntax}, making them valuable tools for developers. One of their primary strengths lies in understanding code and describing behavior. LLMs can analyze and explain the logic of code snippets, breaking down complex functions into more comprehensible explanations for developers. They are also capable of predicting potential runtime behavior \cite{chen2024reasoningruntimebehaviorprogram} based on static code analysis, such as identifying loops, conditionals, and function calls that may impact the execution flow. By simulating execution paths leveraging their neural models, LLMs can infer what might happen under certain conditions, offering valuable insight into code behavior without running it. Furthermore, LLMs excel in reasoning about high-level code behavior \cite{ding2024reasoningplanningllm} by identifying concepts like state changes (e.g., variable updates or branching behavior) and evaluating algorithmic implications, such as performance considerations and potential impacts of code modifications. This makes them useful for predicting time complexity, spotting inefficiencies, and suggesting areas for optimization \cite{yang2024large, zhou2024surveyefficientinferencelarge}.

\begin{figure*}[htbp]
\centering\includegraphics[width=0.9\textwidth]{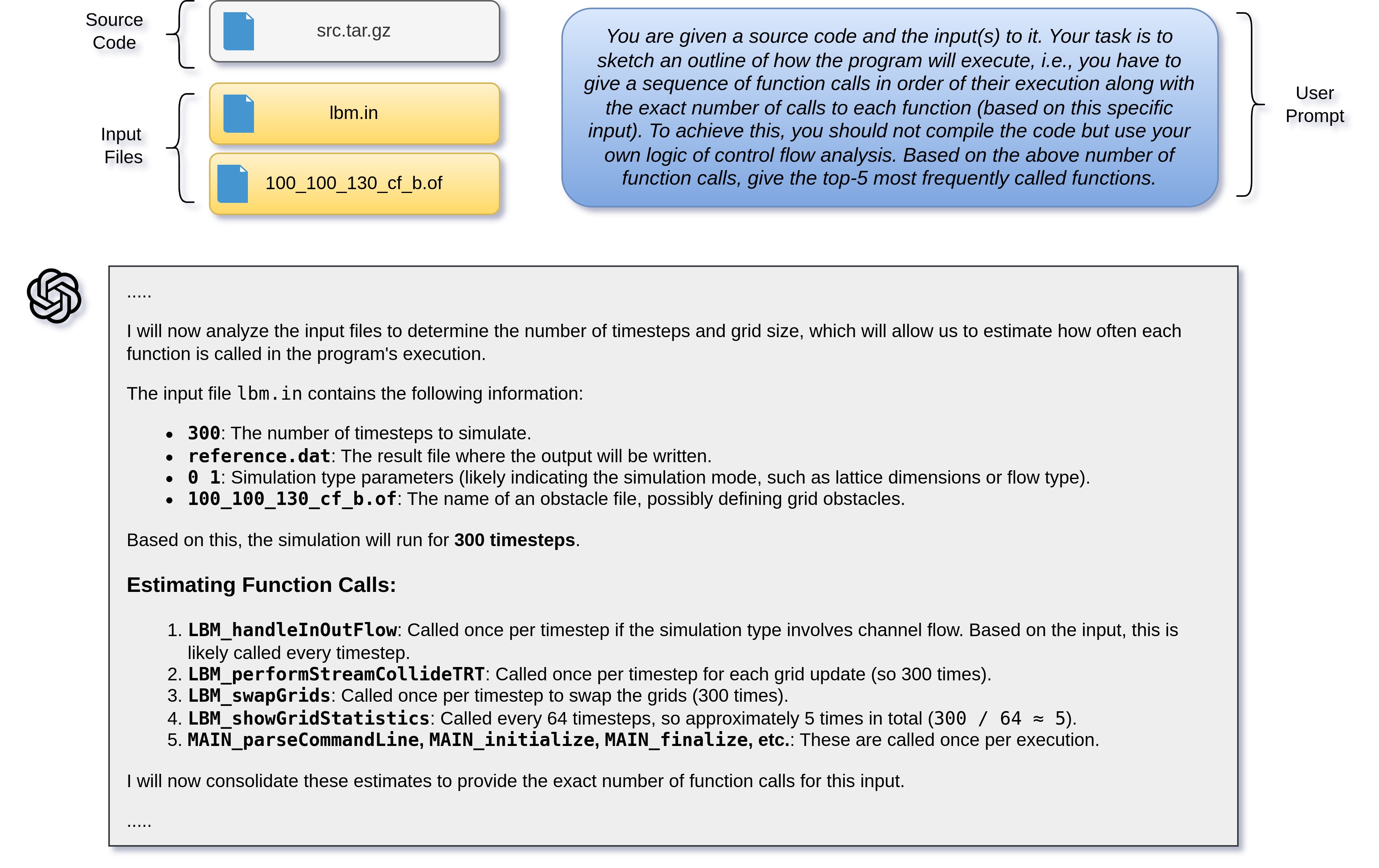}
\caption{An example of successful dynamic program analysis with the LLM, given the source code (\textit{519.lbm\_r}) and input files. The LLM is successfully able to pinpoint \textit{hot} functions in the benchmark for the particular input file, along with the exact number of function calls.}
\label{fig:5}
\end{figure*}

However, LLMs have limitations that stem from their inability to perform deep dynamic analysis. Because they do not execute code \cite{xue2024, yang2024llm}, they cannot analyze real-time state changes, such as tracking memory usage during execution or handling complex input/output scenarios dynamically. Their pre-trained nature also makes it challenging to reason about concurrency and parallelism, which involves understanding multi-threaded behavior, race conditions, or potential deadlocks. This type of analysis \cite{sun2024llmruntimeerrorhandler, chen2024reasoningruntimebehaviorprogram, krishna2024codellmdevkitframeworkcontextualizingcode} requires direct code execution and monitoring of parallel processes, which LLMs cannot achieve. Furthermore, LLMs lack the capability to handle real-time data changes or interact with code based on live data, as they primarily operate on static representations and their training knowledge. This limitation prevents them from fully simulating execution environments or adapting to live input changes, highlighting an important gap compared to runtime analysis tools.

% {\color{red} This last para should talk about exploiting the above, and playing to LLM's strengths - this should make way for the upcoming para}

Although LLMs are inherently limited in performing extensive dynamic analysis, particularly in complex scenarios involving real-time state changes, concurrency, and live data interactions, our framework demonstrates a path to extend their utility through strategic assistance. By combining static analysis techniques with the logical reasoning \cite{parmar2024, li2024relational} and domain knowledge that LLMs inherently possess, we can bridge this gap effectively.

\subsection{Compiler-Assisted LLM Reasoning: Combining World Knowledge with Static Analysis}

Static compiler analysis provides a depth of insight that LLMs relying solely on static inference from code often struggle to match. By meticulously examining the structural components of code, static analysis can prioritize and highlight important functions based on their position in the control flow hierarchy and their interactions across various program components. For example, identifying functions with self or mutually recursive calls or functions nested deeply within loops that span multiple function calls can indicate computational hotspots. This hierarchical understanding, which includes data on call frequencies, loop nesting levels, and data dependencies, offers a precise view of performance-critical regions within a codebase. Such insights are valuable for understanding which functions might contribute the most to the overall execution time of an application.

When these compiler generated static artifacts are paired with an LLM’s capabilities in reasoning \cite{sun2024llmruntimeerrorhandler, chen2024reasoningruntimebehaviorprogram, krishna2024codellmdevkitframeworkcontextualizingcode} about control flow, algorithmic complexity, and domain-specific behavior, a comprehensive analysis of code behavior emerges. The LLM can use static analysis insights to simulate runtime behaviors, assess interactions, and predict the impact of changes more effectively. For instance, a call graph might reveal complex function dependencies, while the LLM can interpret how these dependencies influence runtime performance or suggest optimizations. However, performing such analysis within the LLM itself can be computationally expensive, shifting focus from estimating \textit{hot} functions to more generalized codebase analysis. Since it is an interactive tool, LLMs typically budget a certain time interval for getting the response back to the user and within that budget different phases within the LLM must fit. To offload the burden of analysis and to utilize that time for some other inferencing phases, we assist the LLM with logical cues such as the set of all functions used within the benchmark, function call graphs, and some \textit{hotness} indicators such as whether or not called within loops etc. This approach reduces the model’s search space, enabling more efficient and targeted inference by emphasizing functions that are statistically and structurally significant. This combination of static analysis and probabilistic guidance empowers the LLM to deliver more accurate and efficient code insights, bridging gaps in its inherent capabilities \cite{hadi2024}.

\subsection{Alternative Approach: Profile Generalization \& Generative Models}
\label{sec:back1}

As an lightweight alternative to the inference costs and context length restrictions of LLM, we explore \textit{whole program paths} as a means to reason about application behavior across different inputs. The \textit{Whole Program Path} (WPP) based function profile represents the trace of complete control-flow path exhibited by an application during its execution. The program path, can be recorded either in terms of low-level program entities such as basic blocks and instructions, or in terms of aggregate artifacts such as branch outcomes. In this work, we collect the \textit{WPP} based function profile for the application's dynamic function calls, on a single input. Fig. \ref{fig:0} represents the overall function call and control-flow behavior for a simple benchmark (519.lbm\_r) in SPEC2k17 suite. The corresponding WPP based function profile data is also depicted in the figure. This data can be interpreted as a stream of \textit{tokens}, where each token represents a dynamic function call at a particular timestep in the program execution. For a given execution timestep, the function call token is dependent on the previous state(s) of the execution. For instance, in the given example (Fig. \ref{fig:0}), function call token $7$ (\textit{init\_channel}) is invoked by the another function call $2$ (\textit{main\_initialize}) which occurred seven execution timesteps previously. Thus, in a given WPP based function profile sequence, the prediction of upcoming tokens requires historical information (previous execution timesteps), and involves handling of variable context length.

\begin{figure*}[htbp]
\centering\includegraphics[width=0.9\textwidth]{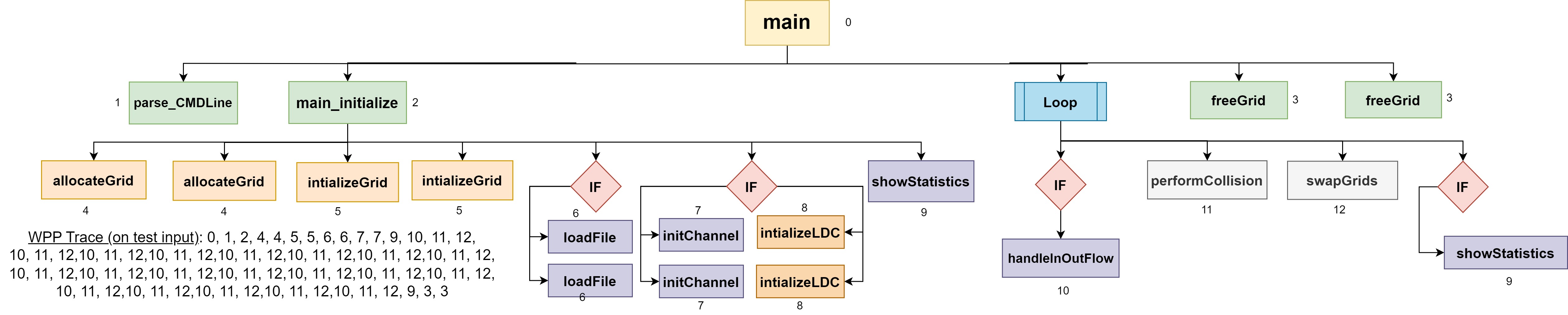}
\caption{A simple outline of control-flow behaviour and function calls in 519.lbm\_r (SPEC2k17). The WPP based function profile data on the small input is also depicted. }
\label{fig:0}
\end{figure*}

In order to utilize this WPP based function profile data for profile generalization and accomplishing profile-driven tasks, we train \textit{generative models} on \textit{WPP} based function profile data, which predicts the upcoming tokens (function calls) based on a prefix of tokens, and simultaneously adjusts to other execution scenarios (different inputs). Deep-learning generative models such as \textit{Recurrent Neural Networks} (RNNs) can estimate the probability of the next token (at timestep $t$) in a sequence, as function of the previous token (at timestep $t-1$) and an internal short-term memory (which propagates information across the sequence). The utility of RNNs stems from the fact that the output token in the previous execution timestep acts as the input token for the current execution timestep, which simplifies the process of querying and inferencing the model. Popular generative models such as GANs \cite{goodfellow2020generative}, Diffusion Models \cite{yang2023diffusion} etc, are extremely compute-heavy ($\sim 10^{20}$ FLOPS per epoch), and require extensive resources for training (100s of GPU hours). In contrast, we show that simple RNNs with one hidden layer (vanilla RNNs) are capable of learning the dynamic function call chains (\S \ref{res:morph}).

However, there are two major challenges in directly adapting such an approach. The first issue is the scalability aspect of WPP  based function profile traces - the disk size of function call WPP  based function profile data ranges from a few GBs to $10^2$, which can lead upto billion of function call tokens (Table \ref{tab1}). For training RNN models, this data needs to be tokenized and encoded into tensors, which makes the dataset size unscalable (Fig. \ref{fig:motiv1a}). 

\textbf{Need for an input-consistent compression scheme}: In order to tackle this scalability issue of WPP  based function profiles, it's essential to compress and remove redundancies. A natural approach for compressing \textit{WPP} based function profiles would be to simply encode frequent substrings of tokens into specific IDs, and then replace their occurrence in the profile accordingly. For example, in Fig. \ref{fig:0}, the token substring $\{10 \to 11 \to 12\}$ inside the loop can be encoded as token $\{10'\}$, reducing the size of WPP based function profile trace. Unfortunately, such substring-based strategies eliminates the control-flow semantics from the WPP based function profile during the training process. For instance, during the model's testing, if the path inside the loop is $\{11 \to 12\}$ (both \textit{if}-statements false), for a different program input, then the compression encoding would not be useful, since substring $\{10 \to 11 \to 12\}$ is absent. Furthermore, the path encoding must be consistent across WPP based function profiles generated on different inputs, since the trained model needs to be tested on unseen, compressed WPP based function profiles. Thus, for successful compaction and generalization of WPP based function profiles, the compression scheme must not be dependent on a specific program input.

\textbf{Incorporating Program Context}: The second challenge arises from the fact that in the its current format, the \textit{WPP} based function profilee does not include the full context of the program. For instance, in the example provided in Fig. \ref{fig:0}, the conditional program path $\{8 \to 8\}$ (\textit{if}-statement false) is not reflected in the \textit{WPP} based function profile trace generated on the small input (\textit{SPEC-test}). In general, the possible program paths which are not expressed in a given \textit{WPP} based function profile, needs to be provided in the training data for the generative model. Without the knowledge of program semantics and possible program paths, the raw \textit{WPP} based function profile cannot be used directly for training generative models, and profile generalization. 

\textit{Morpheus} tackles the first challenge by devising a compression scheme that compacts tokens based on their control-flow behavior. It performs static branch analysis to determine tokens that can be chained together (compressible regions) for path encoding. In order to obtain input-consistent path encoding, \textit{Morpheus} performs the path analysis directly on the source code. For instance, in the above example (Fig. \ref{fig:0}), loops paths $\{8\}$, $\{11 \to 12\}$ and $\{9\}$ would be given separate path encoding, because of the \textit{if}-statements inside the loops. \textit{Morpheus} also performs depth-first traversals on the function-call graph to handle cases with nested conditionals. Secondly, the lack of program context in the \textit{WPP} based function profile is solved by \textit{augmenting possible program paths} in the generated \textit{WPP} based function profile. In the current example, sub-path $\{8 \to 8\}$ is encoded and added in the compressed \textit{WPP} based function profile to provide the complete context of the program to the generative models. The complete workflow of \textit{Morpheus}, including its front-end which is responsible for compression \& augmentation of \textit{WPP} based function profile, and the back-end which trains the generative models, is presented in Fig. \ref{fig:1}.
\begin{figure*}[htbp]
\centering\includegraphics[width=1.0\textwidth]{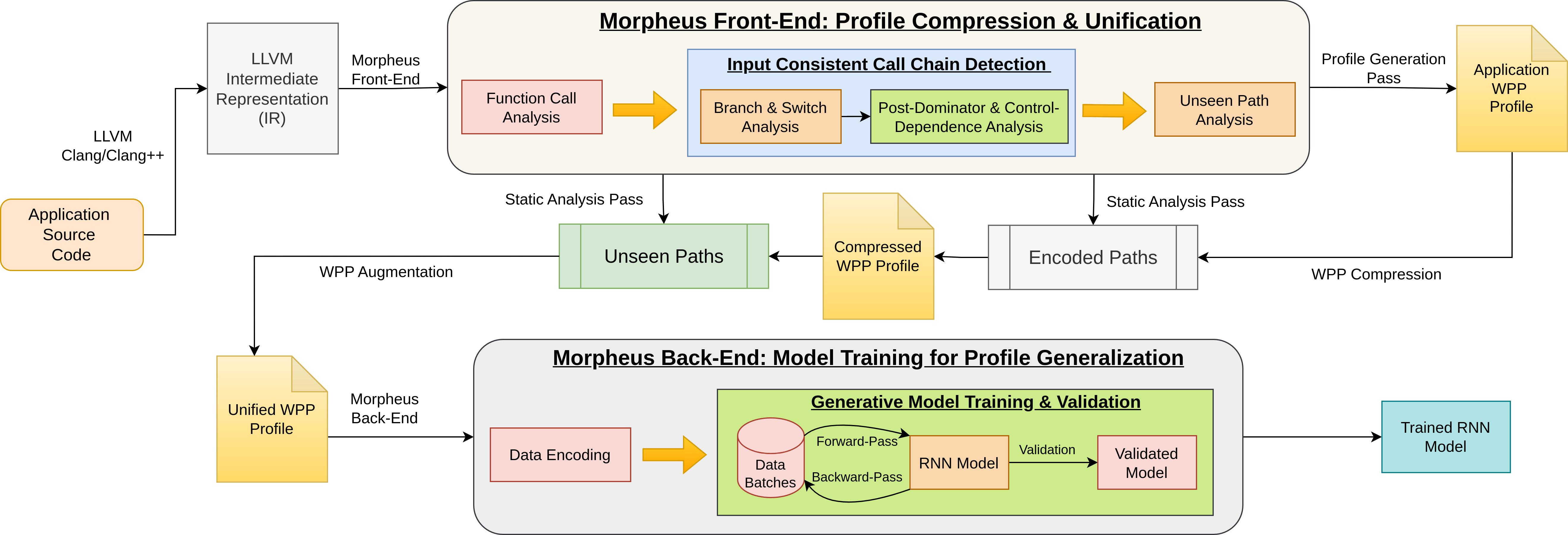}
\caption{\textbf{\textit{Morpheus} Compilation Framework}. The front-end of \textit{Morpheus} is build as integrated set of LLVM compiler passes, while the back-end leverages deep learning framework \textit{Pytorch} to data processing and training the generative model.}
\label{fig:1}
\end{figure*} 

\subsection{Basic Terminologies}
\label{sec:back2}

In this section, we provide a formal definition of two classes of \textit{hot} functions, and some fundamental compiler terminology that is used extensively throughout the paper. Please see appendix A. 

\section{Dynamis: Using LLMs to Predict Program Behavior without Application Profiles}
\label{sec:dynamis}

\begin{figure*}[htbp]
\centering\includegraphics[width=0.9\textwidth]{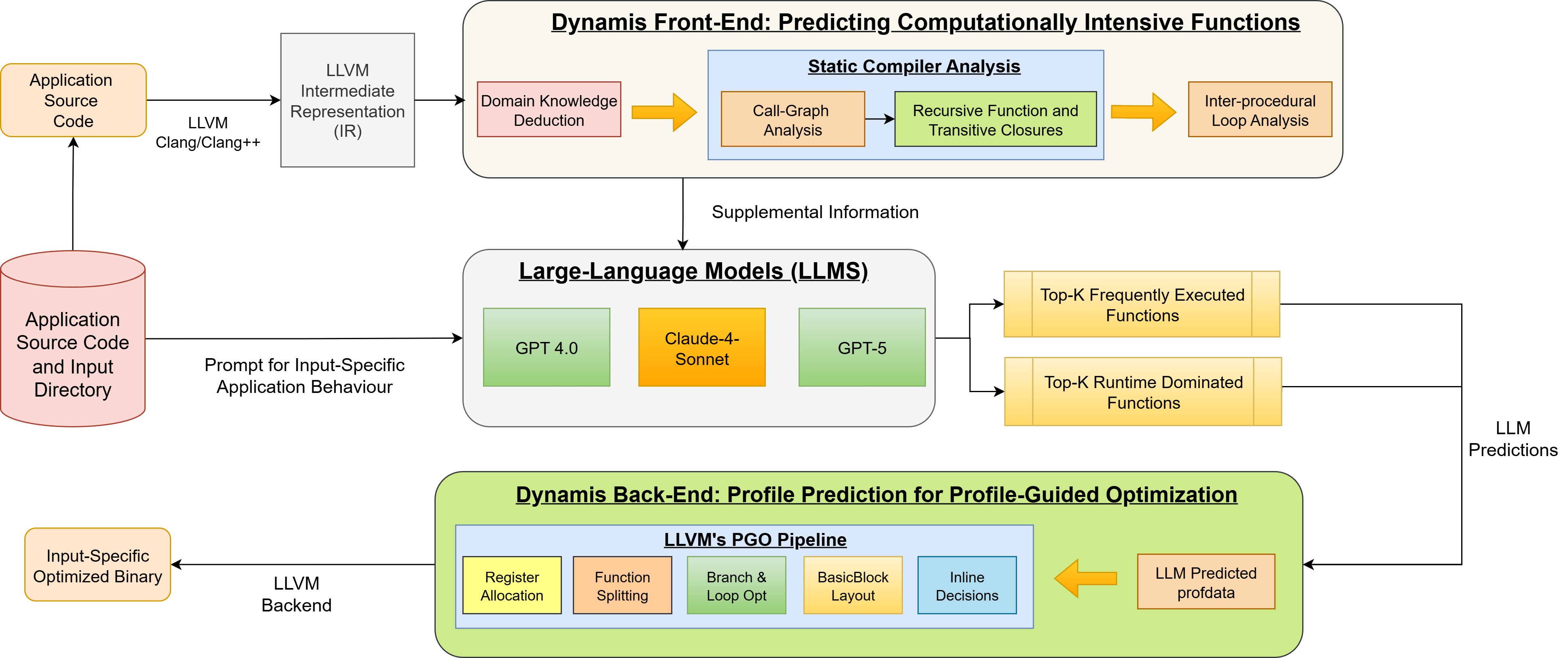}
\caption{\textbf{\textit{Dynamis}: A compiler-assisted LLM framework for predicting dynamic application behavior without any execution}. The front-end of \textit{Dynamis} provides the supplemental compiler analysis required by LLMs to reason about input-specific hot functions, while the backend involves leveraging those predictions to perform optimized compilation using LLVM's profile-guided optimization (PGO) pipeline.}
\label{fig:dynamis}
\end{figure*}

This section introduces \textbf{Dynamis}, compiler-assisted LLM framework which predicts the applications dynamic \textit{hot} functions for a given input, without execution. The front-end of \textit{Dynamis} consists of a set of LLVM-based compiler passes that generates the program artifacts required for LLM reasoning, while the back-end focuses on utilizing these predictions for performing superior profile-guided optimizations. In this paper, we are exploring \textit{Dynamis}' capabilities to determine two independent classes of \textit{hot functions}: \text{most frequently executed hot functions} and \textit{runtime dominated hotspot functions}. 

\subsection{Phase 1: Filtering Search Space with Domain Knowledge Deduction}
\label{sec:domain}

The first challenge in applying LLM-based reasoning across whole applications is navigating the overwhelming search space of thousands of functions, among which certain computation kernels dominate the application runtime. In the absence of additional guidance, it is difficult to distinguish between the actual kernels which are computationally intensive and the driver/utility functions which encapsulate the program hotspots. This makes it challenging for an LLM to isolate the true bottlenecks for optimization. \textit{Dynamis} addresses this challenge by first deducing the domain of the given application (such as \textit{molecular dynamics}, \textit{fluid mechanics}, etc) with prompting. Once the domain is identified, algorithmic cost models and domain-specific principles are applied to predict which operations are likely to be the true bottlenecks. These deductions are then used to filter the effective set of candidate hotspot functions, ensuring that subsequent analysis focuses on the most computationally relevant regions of the code.

\noindent\begin{minipage}{.45\textwidth}
\begin{lstlisting}[language=C,caption={Finite Element Solver code snippet from 510.parest\_r in SPEC2017 \cite{bucek2018spec}. While the wrapper dominates inclusive time, the true hotspots are the inner quadrature and scatter kernels, based on domain knowledge.},label=code1,frame=tlrb]
void assemble_matrix() {
 for (cell : mesh) {
  for (q_point : quad_points) {
   double phi  = fe_values.s_val(.);
   Tensor grad = fe_values.s_grad(.);
   l_mat(i,j) += phi * grad;}
   cpy_matrix(l_mat, g_mat);}}
\end{lstlisting}
\end{minipage}\hfill
\begin{minipage}{.45\textwidth}
\begin{lstlisting}[language=C,caption={Primal Simplex code snippet from solving minimum cost flow in 505.mcf\_r in SPEC2017 \cite{bucek2018spec}. Despite enclosing the optimization loop, the runtime is dominated by the inner pricing and sorting routines},label=code2,frame=tlrb]
void primal_net_simplex() {
 while (!optimal) {
   primal_bea_mpp();      
   spec_qsort(.., cost_compare);
   update_tree(); 
   price_out_impl(); }}
\end{lstlisting}
\end{minipage}

Our initial experiments showed that once the application domain is established, a direct correlation can be established between algorithm-specific cost models, and candidate hotspot functions. For instance, $510.parest\_r$ from SPEC2k17 \cite{bucek2018spec} implements a \textit{finite element solver} for biomedical imaging (Listing \ref{code1}). Based on this domain knowledge, the LLM can immediately deduce that quadrature evaluations, gradient computations, and scatter operations are computationally intensive. This leads to the isolation of \textit{s\_val}, \textit{s\_grad}, and \textit{cpy\_matrix} as potential hotspots, sidestepping the high-level \textit{assemble\_matrix} wrapper. Similarly, $505.mcf\_r$ from the same benchmark suite adopts a \textit{primal network simplex algorithm} to solve the \textit{minimum cost flow problem} (Listing \ref{code2}). This deduction enlightened the LLM to attribute runtime to the domain specific constructs such as pricing (\textit{cost\_compare}), sorting (\textit{spec\_qsort}), and implicit-arc maintenance (\textit{bea\_is\_dual\_infeasible}) functions that dominate execution at scale, instead of the top-level driver \textit{primal\_net\_simplex}. We observed a similar pattern in other application domains, such as \textit{molecular dynamics} (neighbor list construction and electrostatics kernels), \textit{fluid mechanics} (collide-and-stream sweeps), and \textit{image processing} (convolution, and resampling), where domain knowledge consistently filtered out high-level drivers in favor of the computational kernels.

\begin{comment}
\begin{tcolorbox}[enhanced,frame style image=blueshade.png,
  opacityback=0.75,opacitybacktitle=0.25,
  colback=blue!5!white,colframe=blue!75!black,
  title=Prompt for leveraging LLM's world knowledge for deducing application domain]
  \textit{Using the code structure of the application or through your world knowledge about the application <name>, determine and report the problem solved by the application and the typical algorithms used by the benchmark. Also, determine <K> most frequently used functions in the benchmark.}
\end{tcolorbox}
\end{comment}

Furthermore, leveraging application domain knowledge allows \textit{Dynamis} to apply algorithm-specific cost models, such as $O(nz)$ for sparse matrix–vector products, $O(N^{2})$ for neighbor-list generation, and $O(N \cdot r^{2})$ for convolutional operators, etc. This enables \textit{Dynamis} to reason about scaling behavior and identify algorithmic bottlenecks across a wide range of applications. Note that the domain itself is inferred from several programmatic constructs, such as file structures, function signatures, symbolic constants, algorithmic patterns, and control-flow context. To validate this, we performed experiments with obfuscated function names and removed comments - the LLM still managed to infer the right domain by relying on algorithmic structure and control-flow regularities. 

\subsection{Phase 2: Supplementing with Program Semantic/Structural Artifacts}
\label{sec:semantic}

The second challenge in leveraging LLMs for arises when there is an ambiguity between the `moderately hot' functions and dormant functions, which do not contribute to the program execution. Optimizing an application often requires a fine-grained view of program's execution - from the dominant hotspots to secondary routines that still contribute meaningfully to runtime. The inclusion of domain knowledge narrows the space the to plausible kernels, but it can lead to overemphasizing a small set of domain-specific kernels or overlooking moderately expensive functions that accumulate cost across iterations. To address this, \textit{Dynamis} augments its analysis with semantic and structural artifacts derived from static analysis, which includes complete function list, the program call graph, and the set of recursive functions with their transitive closures. 

\begin{wrapfigure}{r}{0.45\textwidth}
\vspace{-8pt}
  \begin{lstlisting}[language=C,caption={Simplified Code snippet from imagick in SPEC2017 \cite{bucek2018spec}. The call-graph traversal from hot kernels allows \textit{Dynamis} to isolate `moderately' hot functions},label=code3,frame=tlrb]
void ResizeImage(Image *img) {
 HorFilter(img); // medium hot
 VertFilter(img); // medium hot
 SyncCacheViewAutPix(); // cold
}
\end{lstlisting}
  \vspace{-18pt}
\end{wrapfigure}

The call-graph sheds light on caller–callee relationships, clarifying how the apparent cost of high-level drivers is distributed among their callees, and distinguishing moderately hot kernels that accumulate significant work from functions that remain largely dormant. This relationship is further refined with the inclusion of recursive functions, and helps the LLM to account for \textit{call multiplicity}, which essentially identifies functions due to repeated invocation along multiple paths. For instance, Listing \ref{code3} shows a simplified code snippet from $538.imagick\_r$, where the call-graph highlights how filter kernels (\textit{HorFilter} and \textit{VertFilter}) and cache operation (\textit{SyncCacheViewAutPix}) from the domain specific hotspot function \textit{ResizeImage}. In this case, the both the filter operations are moderately hot for the large input, while the cache sync remains relatively insignificant. By traversing these caller–callee relationships, the call graph provides a graded view of functions—from dominant hotspots to moderately hot kernels to dormant routines, that better guides downstream optimization opportunities. 

\begin{comment}
\begin{tcolorbox}[enhanced,frame style image=blueshade.png,
  opacityback=0.75,opacitybacktitle=0.25,
  colback=blue!5!white,colframe=blue!75!black,
  title=Prompt for determining Top-K most frequently executing function for a given input]
  \textit{Now, determine the <K> most frequently executed functions in the application <name> for the given input <input-file>. For supplementing the LLM knowledge, the following semantic information are provided: (i) the list of all functions present in the application, (ii) the call graph of functions present in the application (adjacency list), (iii) the set of recursive functions, and their transitive closure.} 
\end{tcolorbox}
\end{comment}

\subsection{Phase 3: Fortifying LLM with Static Compiler Analysis}
\label{sec:staticLLM}

The addition of function call graphs and the recursive functions helps the LLM narrow down the function hotspots through caller–callee relationship, but does not quantify how often a kernel is actually exercised across different execution paths. For instance, two callee functions may appear equally central in the call-graph, but their effective runtime contributions can diverge sharply across different iteration of a loop nest. To address this issue, \textit{Dynamis} augments the LLM with a static measure of \textit{interprocedural loop depth}, which tracks how deeply a function is nested in loops both within its own body and across its transitive callers (definition \ref{def:interloopiness}). This provides an architecture-agnostic way to estimate repetition: functions invoked at deeper nesting levels or through recursive structures are more likely to dominate execution time. 

\textit{Dynamis}' algorithm to estimate the inter-procedural loop depth is presented in the appendix (Algorithm A1). It computes the interprocedural loop nesting depth of functions by first recording intra-procedural nesting at each call site, then propagating these depths across the call graph using a longest-path calculation from the program’s main entry point. In practice, our experiments showed that combining all three phases puts \textit{Dynamis} in the best position to determine computation hotspots of an application. For instance, in $519.lbm\_r$ from SPEC CPU2017, a lattice–Boltzmann solver, the repeated full-grid sweeps consistently surfaced the collide–stream kernel as the dominant hotspot, with interprocedural loop depth highlighting its repeated execution across timesteps. By weighting call-graph edges with loop depth, \textit{Dynamis} was able to filter out shallow wrappers and library calls, while consistently surfacing the kernels that dominate execution across specific input sets.

\subsection{Phase 4: Performing Profile-Guided Optimization}

Equipped with the supplemental materials from the earlier phases, semantic program artifacts (function lists, call graphs, recursion sets) and static analysis, \textit{Dynamis} prompts the LLM to generate a refined set of top-k computationally hot functions for the application under a given input. As mentioned earlier, \textit{Dynamis} generated two classes of independent predictions: \textit{Top-k frequently executed functions} and \textit{Top-k runtime dominated hotspot functions}. Both classes of predictions are required for contrasting the LLM's predictive power, since they serve different optimization goals—one characterizes control-flow frequency, while the other identifies performance-critical bottlenecks.

After generating these predictions, \textit{Dynamis} translates them into LLVM \textit{profdata} files, which contains function hashes, basic block execution counts, branch probabilities in an indexed profile format. \textit{Dynamis} searches for each hash corresponding to a predicted hot function, and then manipulates the counters based on the function ranking. This LLM-modified \textit{profdata} is then passed onto LLVM's profile-guided optimization framework, where it acts as a guideline for several optimization decisions such as function inlining, basic block placement (block layout), branch probability optimization, loop optimizations (unrolling, vectorization and peeling), function splitting (hot/cold splitting), register allocation and speculative optimization. \textit{Dynamis}'\textit{profdata} file ensures that the functions predicted by the LLM are boosted during PGO in the order of their predicted hotness.

%\textbf{Leveraging Hot Functions for Profile-Guided Optimization (PGO)}: The compiler leverages the predicted top-k hot functions in form of order files, which explicitly specifies the desired sequence for functions within the final executable. It serves as a blueprint for the linker to arrange functions in memory for optimal binary size and performance \cite{hoag2024reordering}. The compiler performs code layout optimization through function reordering, where hot functions are placed contiguously in memory using order files, improving instruction cache locality and reducing page faults. This reordering also helps in optimizing the application startup time, by ensuring that the functions needed at startup are loaded in fewer memory pages, thereby reducing page faults. 

\section{Morpheus: Generalizing Whole Program Path Profiles}
\label{sec:comp}

This section introduces the \textbf{Morpheus Compiler Framework}, which predicts the application dynamic function calls by generalizing its whole program path profiles (Fig. \ref{fig:1}). The front-end of \textit{Morpheus} is a LLVM-based instrumentation framework that generates the \textit{WPP based function profile} in a compacted format, and performs static analysis to identify compressible regions in the \textit{WPP} profile, and determines the paths for incorporating program context. The compacted WPP based function profile is generated by instrumenting the program loops, and the performing various online compression strategies. This allows \textit{Morpheus} to replace the redundant function call sequences, which in turn, reduces the bias in WPP based function profiles. The static analysis for determining compressible regions is performed at the intermediate-representation (IR) level in order to obtain a compression scheme $\{path:code\}$, which can be applied on profiles generated by any program input. We refer to this scheme as \textit{input-consistent} compression scheme. 

After applying the compression scheme, the unseen program paths are augmented to generate a compressed, unified version of the WPP based function profile. This unified WPP based function profile serves as the training data for the \textit{Morpheus}'s back-end. The profile is encoded into tensors, and then used to train lightweight generative models. The trained generative model can be leveraged to predict the application's behavior on unseen inputs.

\subsection{Collecting Whole Program Profiles: Smart Loop Instrumentation}
\label{sec:loop}

The primary reason behind the intractability of WPP based function profiles in terms of their size, is the presence of highly redundant tokens and sub-paths. These redundancies are the result of the iterative behavior (deeply nested loops) of the program during its execution. Such data redundancies create extreme imbalances in the token distribution, which further causes significant bias in the WPP based function profile. This inhibits the downstream task of training generative models on the WPP based function profile. \textit{Morpheus} tackles this problem by performing \textit{smart-loop-instrumentation} and \textit{online compaction}, during the generation of WPP based function profiles. The intuition here is that instead of directly logging function tokens in the WPP based function profile during nested loop executions, they are first compacted and summarized into lossless representations, while preserving their relative order in the program context.

\textit{Morpheus} implements the \textit{smart-loop-instrumentation} scheme by first instrumenting the loop headers, and the loops exit blocks, which trigger the activation and de-activation of a loop phase, with the help of special flags. Any function call tokens produced during loop  phase are then stored into a temporary buffer. The contents of this buffer are subjected to online compaction whenever all loop phases in the application are deactivated. The online compaction scheme is described in Appendix (Algorithm 2). The distribution of the temporary buffers is first analyzed to see if it has been observed previously (lines 3 - 6). For compacting the buffer contents, the longest, repeating, unique and non-overlapping function call sequence is obtained, by designing a dynamic programming-based approach. This call sequence is then removed from buffer, and a special priority number is attached to the call sequence . The purpose of this counter is to determine the relative order of the longest call sequences. The algorithm is repeatedly invoked in the temporary buffer, until the content of the buffer are empty. Each longest, repeating call sequence is then logged in the WPP based function profile (in their respective order), with a begin and end markers to indicate the loop phase.

\subsection{Identifying Compressible Regions: Static Program Analysis}
\label{sec:static}

The compressible regions in the WPP based function profile represents a sequence of function calls that can be grouped together into a single call. In absence of any control-flow (program branch), which can lead to variant program behavior, the entire sequence of function call tokens can be replaced with a single token. A naive approach to achieve this would be to simply group tokens which are not affected by any program branch. However, such an approach would miss compression opportunities for tokens which are nested inside branches. For determining the compressible regions in the WPP based function profile, \textit{Morpheus} leverages the following insight regarding the function calls and branch targets, which can be stated and proved as a lemma:

\begin{lemma}\label{l1}
    In the absence of inter-procedural branches, function calls that are control-dependent on the same branch, and are present on the path of same branch target, can be compressed into a single token. 
\end{lemma}
\begin{proof}
Please see appendix C1. 
    %Let $B$ be a program branch, and $F_1, F_2$ be two function calls that are control-dependent on $B$. By definition of control-dependence, we have a path $p_1: B \to F_1$, where $F_1$ post-dominates all nodes in $p_1$. Similarly, there exists another path $p_2: B \to F_2$, where $F_2$ post-dominates all nodes in $p2$. If both $F_1$ and $F_2$ are present on the path with same branch target $T$, then $p_1 \equiv p_2$. This means that regardless of the input, if the program execution follows path $p_1$, then $F_1$ and $F_2$ are statically guaranteed to appear in the the trace of the path(not necessarily next to each other but either one followed by the other). In this case $F_1, F_2$, and any intermediate function calls appearing between these calls, can be compressed into a single token.
\end{proof}

In order to determine control-dependence relationships, \textit{Morpheus} first analyzes all function calls, branches (including \textit{switches}) in the source code, and constructs the \textit{post-dominator} tree which is used for detecting control dependence. The control dependence of each function call is checked with both the nearest branch and switch statements.

For simultaneously computing the control-dependence between function calls and branches, and identifying the specific branch targets on which the control-dependence property holds, \textit{Morpheus} implements a \textit{bit-vector based detection mechanism} (Algorithm 4). The algorithm treats each branch target as an element in a bit-vector (line 5). If a function call post-dominates the $i^{th}$ branch target, then position $i$ of the bit vector is set to 1  (lines 10 - 11). Thus, if a function call post-dominates all (or none) targets of a branch, then the \textit{bitwise exclusive-OR} (lines 15 - 20) of the bit-vector would be 0. In that case, the function call is not control-dependent on the branch (lines 24 - 26). Otherwise, the function call dominates atleast one branch target (but not all), and is control-dependent on the branch (lines 21- 23). The output of Algorithm 4 is a mapping between \textit{every function call} in the program and the \textit{corresponding branch target} on which they are control dependent on (control independence is denoted as $0$).  

\subsection{Input-Consistent Compression}
\label{sec:comp1}

The compressible regions in \textit{WPP based function profile} profile can be identified by observing control-dependent function calls and the corresponding branch targets which they post-dominate. However, lemma \ref{l1} assumes the absence of any control-flow branches, which can potentially be present in callee functions, called by control-dependent functions. In order to reason about \textit{interprocedural branches} for compression, \textit{Morpheus} performs a modified depth-first traversal of functions, starting from the \textit{main} function. This process is highlighted in Algorithm (Appendix C). For every intermediate callee function called by the caller function, this traversal takes into account the branch targets of caller functions, and checks if the callee is further dependent on some other branch. This processes is repeated recursively for all caller-callee functions (lines 8-9). After handling inter-procedural branches, \textit{Morpheus} then checks the regions of the call stack which are dependent on the same branch targets, and encloses them into a single path. 

The path encoding are then matched with sub-paths in the given \textit{WPP based function profile}, where each instances of the sub-path are replaced with the corresponding path code. Our experiments have shown that in practice, this does not affect the compression drastically. The \textit{input-consistency} of this compression scheme results from the establishing the path encoding directly from the application intermediate representation, instead of analyzing individual program profiles.

%To minimize the computation time, \textit{Morpheus} utilizes a greedy maximal string matching algorithm with polynomial time complexity ($\bigcirc (m*n)$), since this problem is similar to the \textit{optimal word break problem} \cite{wbreak}, which has exponential time-complexity ($\bigcirc (2^n)$)

\subsection{Augmenting Whole Program Profiles}
\label{sec:aug}

\begin{wrapfigure}{r}{0.45\textwidth}
\vspace{-20pt}
  \begin{center}
    \includegraphics[width=0.45\columnwidth]{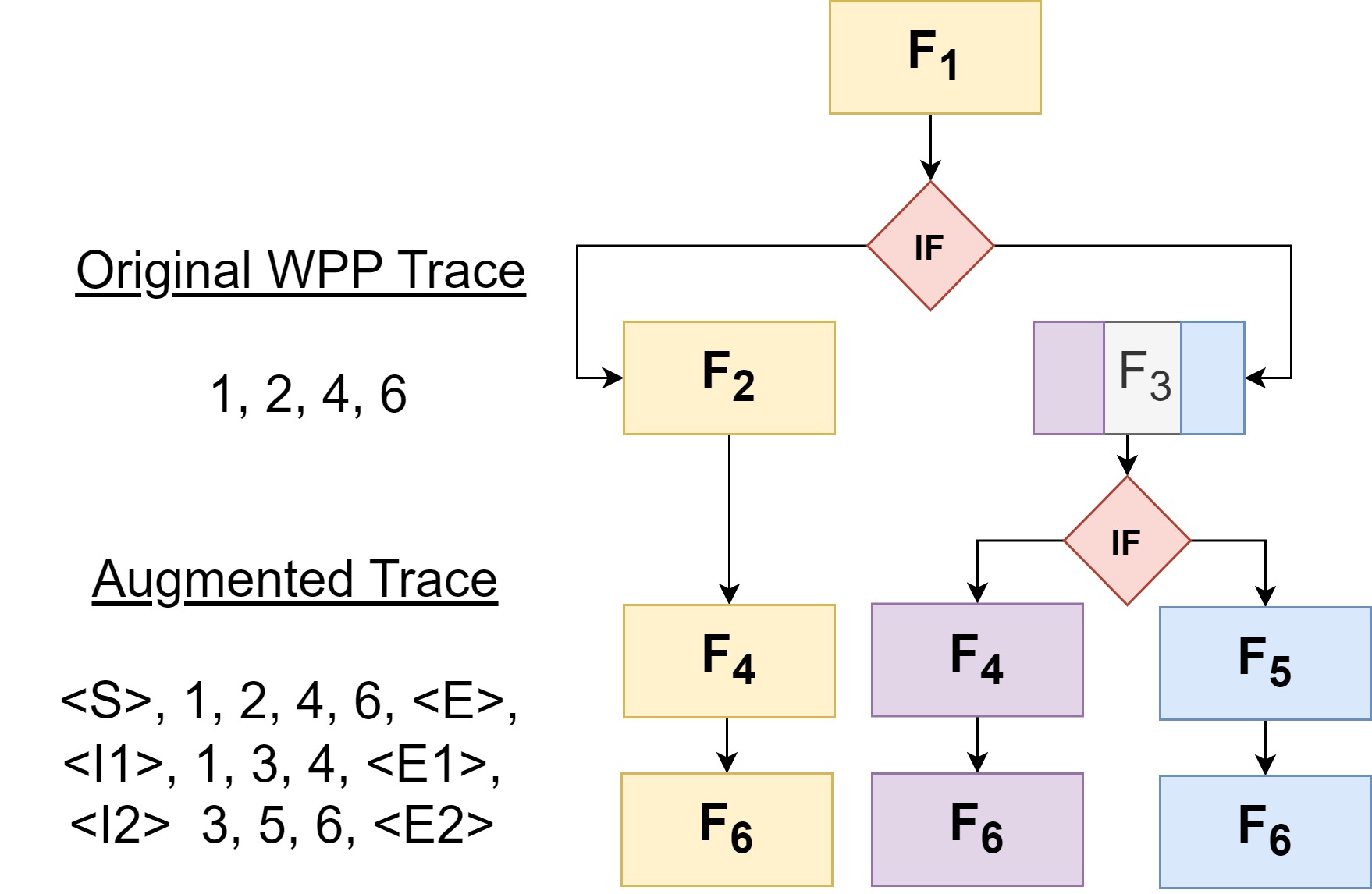}
\caption{\small Simple example of augmenting unseen program paths in the WPP based function profile. The paths containing the set of unseen function $\{F_4, F_5\}$, are obtained from traversing the call-graph starting from each unseen node. The traversal stops at the first seen node.}
\label{fig:aug1}
  \end{center}
  \vspace{-15pt}
\end{wrapfigure}

The final step in the creation of an unified, compressed profile is the inclusion of program context in compressed profile. A single WPP based function profile represents the control-flow paths taken by the program for a particular input. It inherently lacks the knowledge of program branches, and conditionals, and the program context, which can lead to diverse program behavior. \textit{Morpheus} tackles this issue by performing path analysis, and then augmenting the unseen paths in the compressed profile. 

To achieve this, \textit{Morpheus} first statically obtains a list of all the functions which are present in the application. For a given WPP based function profile, it then computes the set difference between the static function list and the observed functions in the WPP based function profile. For each (unseen) function in the set difference, it then performs a depth-first traversal of the function callgraph. The traversal terminates at the first function that is already present in the WPP based function profile. This operation results in a chain of unseen functions, which are then appended at the end of the compressed WPP based function profile. Fig. \ref{fig:aug1} depicts a single example where the WPP based function profile is augmented with the unseen program paths. Each augmented path is appended with \textit{start} and \textit{end} marker tokens, which preserves the context of the original trace. These marker tokens are also intended for indicating the generative model to adapt to input-specific path behavior.

\subsection{Deducing Dynamic Function Calls: Sequence Prediction Problem}
\label{sec:func}

Dynamic function call prediction involves computing the sequence of function calls during a program execution. Specifically, the goal is to develop a prediction mechanism that can determine the function call $F_{t+1}$ at a timestep $t + 1$, based on the program's program execution state $h_t$, at timestep $t$. This can be expressed as sequence prediction problem with learnable weights $\mathbb{W}$, and bias $\mathbb{B}$:

\begin{equation}
   \underbrace{F_{t+1}}_\text{next call} = \mathbb{W} \underbrace{h_{t}}_\text{program state} + ~~\mathbb{B} 
\end{equation}

For predicting the program's dynamic function calls, the program's execution state at timestep $t+1$, can be captured as a combination of function call $F_t$ at previous timestep $t$, and the program's historical execution context (function calls $F_0, F_1, ... , F_{t-1}$). The program state $h_{t+1}$ can then be expressed the following recurrence relation $f_W(h_t, F_t)$:

\begin{equation}
  h_{t+1} = \underbrace{f_W}_\text{activation function}(h_{t}, ~~F_t) 
\end{equation}

In Deep Learning theory, the function $f_W$, which models the next program state based on previous state and function call, is known as \textit{activation function}. \textit{Morpheus} uses a rectified linear-unit (ReLU) as its activation function. The next function call token $F_{t+1}$, and the program state $h_{t+1}$ at timestep $t+1$ can thus be estimated with the help of \textit{Recurrent Neural Network} (RNN). A typical RNN consists of neurons which can be distributed into \textit{input} layer, \textit{output} layer, and \textit{hidden} layers. The function call tokens are first encoded into vectors, and passed on to the input layer. The hidden layer is responsible for modelling the program state by storing relevant information from previous timesteps. In this work, \textit{Morpheus} uses RNNs with a single hidden layer, which is one of simplest neural architecture in theory (\textit{vanilla RNN}).

The utility of RNN also stems from the fact that can be adapted to predict function call sequences in a variety of scenarios. For instance, in case of \textit{one-to-one} prediction scenarios, the RNN accepts a new token at a single timestep, and outputs a single token. For next timestep, this output token can be re-used as an input token, which can give it \textbf{one-to-many} flavor, meaning that with multiple tokens could be obtained from RNN with a single input token. For function call predictions, the single input token would be the token corresponding \textit{main} function of the application.     
\section{Evaluation}
\label{sec:eval}

The evaluation of \textit{Phaedrus} is two-fold. The first set of experiments are designed to evaluate \textit{Dynamis}'s \textit{application behavior synthesis}, while the second set of experiments evaluate the application profile generalization approach by \textit{Morpheus}. Specifically, we evaluate \textit{Phaedrus} to answer the following set of research questions:

\begin{itemize}
    \item Can \textit{Dynamis}'s LLM-based generative model predict most frequently executed dynamic functions for given inputs? Can \textit{Dynamis} be extended to predict the runtime dominated hotspot functions? How do these predictions impact performance and code size?
    \item What sources of additional information should be provided to \textit{Dynamis} for achieving close-to-ground truth results? Finally, what fraction of function calls do the predicted functions cover compared to the ground truth when provided with additional inputs?
 \item How effective is \textit{Morpheus}'s smart-loop instrumentation mechanism in compressing the WPP based function profiles? Is it able to reduce redundancies in a lossless manner?
    \item Can \textit{Morpheus}'s RNN-based generative model predict most frequently executed dynamic functions for unseen inputs? How ``lightweight'' is the entire approach, in terms of training time and model size? 
    \item How much performance improvement \& binary size reduction can be achieved by profile-guided optimization by using function hot sets generated by \textit{Morpheus} \& \textit{Dynamis}? What are the tradeoffs between these two approaches? 
\end{itemize}

\textbf{Experimental Setup}: The experiments were conducted on a Dell 7820 Precision Tower with Intel Xeon Gold 6230R processors. The Intel system has 2 sockets (13 cores each), 36 MB L3-cache, running Ubuntu 20.04.6 LTS. The compiler passes and algorithms in front-end of \textit{Morpheus} \& \textit{Dynamis} were implemented in LLVM 15.0.0, while Pytorch version 2.4.0, was used in \textit{Morpheus}'s deep-learning back-end with Python 3.8.10. GPT-5,  Claude-4-Sonnet and GPT 4.0 were used as the LLM in \textit{Dynamis}. The LLMs were invoked in an interactive environment by \textit{Cursor} \cite{cursor_ai_2024} .   

\textbf{Hyper-Parameters \& Baselines}: The RNN-based model used in \textit{Morpheus} consists of single-layer, with 1000 hidden units. The training batch size was set to 4, 32, and 64 for benchmarks with $\leq50, \leq1000, >1000$ tokens, respectively. The learning rate was set to be $0.001$, dropout $0.2$, and each RNN was trained for 20 epochs.

\textbf{PGO Baseline}: To compare the PGO benefits of \textit{Phaedrus}, we use two different baselines (1) \textit{O3-baseline}: applications compiled normally with O3-flag (2) \textit{traditional PGO}, where an application is first profiled on representative input(s), and then compiler transformations are performed based on the given profile. For this baseline, we profile the application on small input and use that profile to perform PGO.

\subsection{Evaluating Dynamis: Application Behavior Synthesis}
\label{res:dyn}

Our experiment of extrapolating the application domain by \textit{Dynamis} showed a very high (almost perfect) success. With \textit{GPT 4.0}, we observed that except \textit{510.parest\_r}, the deduction is very close to the actual domain. However, when this experiment was repeated with GPT-5 and Claude-4-Sonnet, we observe that the domain deduction is always correct. This makes way for accurate reasoning of the algorithmic behavior of domain-specific applications. Table 1 in appendix shows the complete extrapolation of different domains by GPT 4.0-equipped \textit{Dynamis} in Phase 1 (\S \ref{sec:domain}) across different benchmarks. 

The effects of prompting the LLM to first infer domain knowledge is highlighted in Fig. \ref{fig:phaseres}. The most notable improvement was in benchmark \textit{531.deepsjeng\_r} improvement, which immediately led to the identification of computationally intensive functions. 

\begin{figure*}[ht]
\begin{tabular}{p{0.5\textwidth}p{0.5\textwidth}}
    \begin{minipage}{.49\textwidth}
    \centering\includegraphics[width=1.0\textwidth]{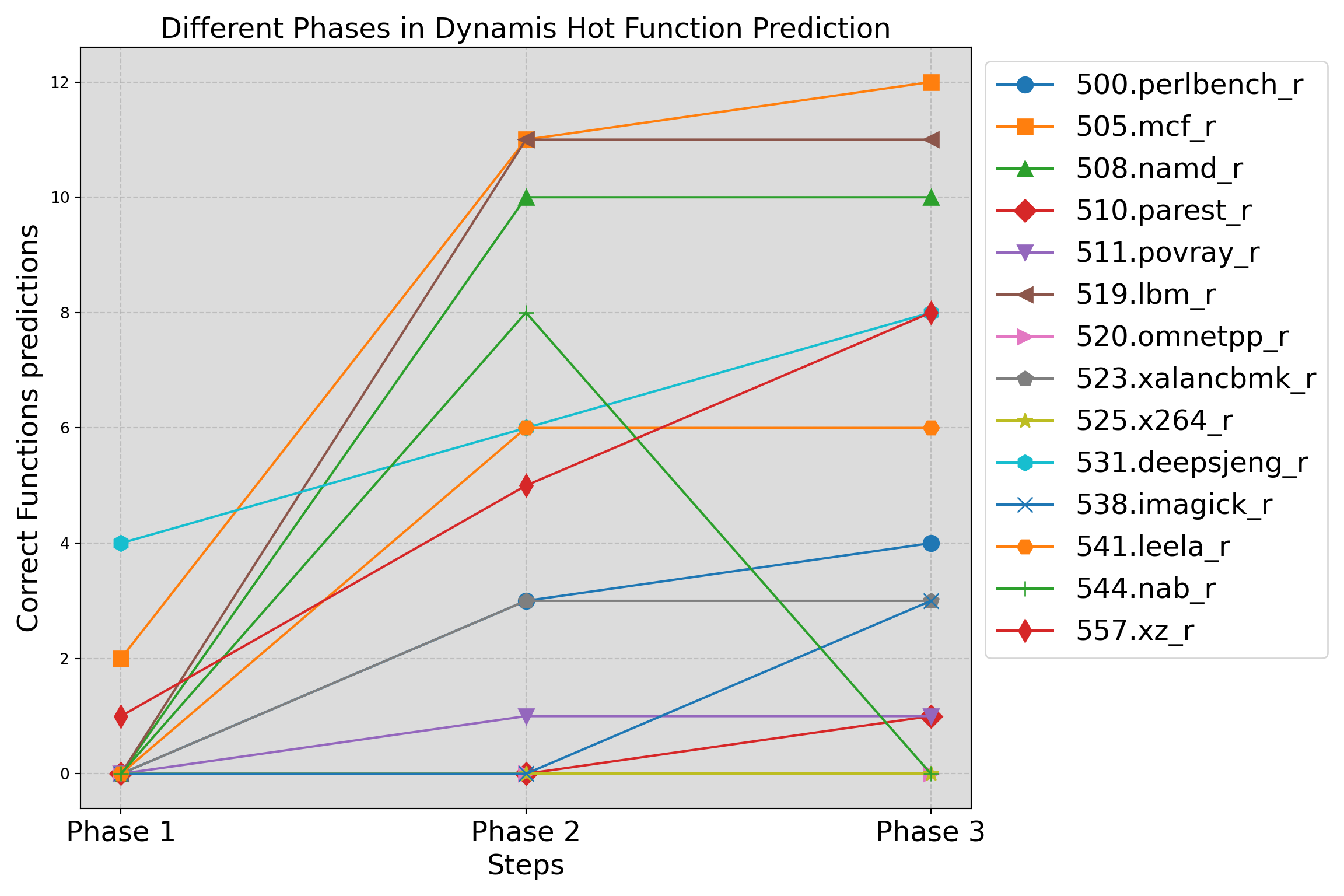}
\caption{\textit{Dynamis}'s prediction of top-15 \textit{most frequently executed} functions in the three phases. The trend how \textit{Dynamis} refines it's predictions, once equipped with finer granularity of compiler knowledge.} 
\label{fig:phaseres}
    \end{minipage}
    &
    \begin{minipage}{.49\textwidth}
    \centering\includegraphics[width=1.0\textwidth]{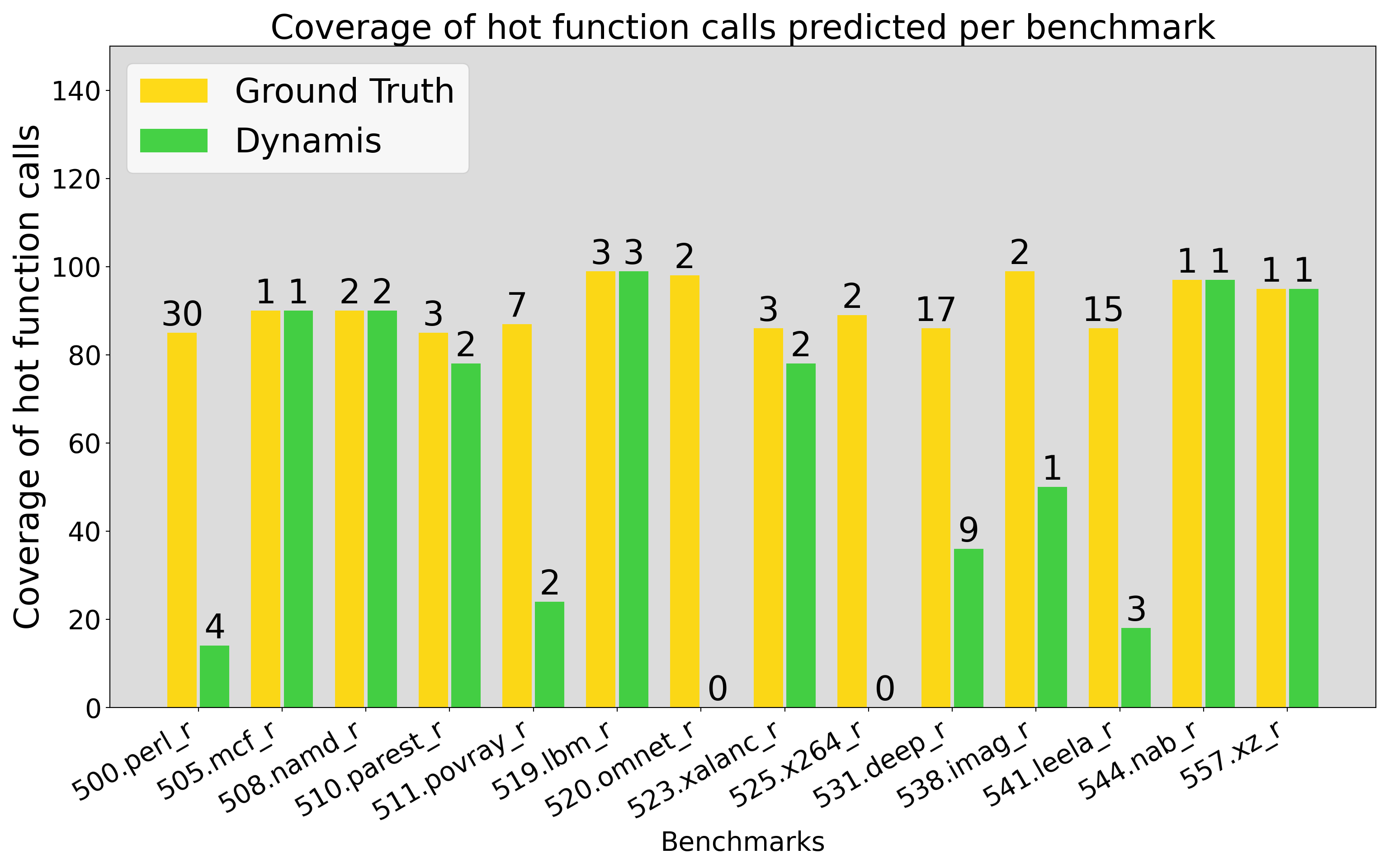}
\caption{Percentage of most frequently executed function calls predicted by \textit{Dynamis}, and compared with the actual execution (ground truth).} 
    \label{fig:funccoverage}
\end{minipage}
\end{tabular}
\end{figure*}

The predictions of \textit{most frequently executed} hot functions improved significantly in phase 2 (\S \ref{sec:semantic}), where the LLM was provided with program semantic artifacts. Specifically, providing function universe improved the results by 8.33\%, notably in benchmarks \textit{544.nab\_r, 508.namd\_r} by 43.33\%, while call graph, recursive functions and transitive closures enhanced the output by 2.8\%, 5.56\%, and 12.22\% respectively. The function universe reduces the search space, while call graphs exposes caller-callee relationships, aiding in statically tracing control flow. We observed that recursive functions also lead to high-execution hotspots and transitive closures uncover long-range dependencies. These elements enhanced \textit{Dynamis}' ability to prioritize and predict computationally intensive functions.

Finally, in phase 3 (\S\ref{sec:staticLLM}),  we observe inter-procedural loop depth, further refines LLM estimates by indicating function \textit{hotness} based on computational complexity. However, estimating the size of loop bound to a particular depth remains challenging. Although \textit{Dynamis} works well on most benchmarks, limitations persist for \textit{520.omnetpp\_r} and \textit{525.x264\_r}, which lack correlations between static analysis cues and runtime behavior, severely making it harder for LLM to determine the most frequently executed functions. 

\subsection{Expanding Dynamis: Exploring Modularity with new LLMs}
\label{res:hotpotpred}
In order to investigate \textit{Dynamis}' performance with different LLMs, we repeated the experiments with two current state-of-the art "thinking" LLMs: \textit{claude-4-sonnet} and \textit{GPT-5}. In these experiments, we designed the prompt to account for the "\textbf{functions where the applications spend the most amount of execution time}". This is different from our previous experiments, since we are shifting the focus from the set of \textit{most frequently executed hot functions} to \textit{runtime dominated hotspot functions}. Predicting runtime-dominated hotspot functions is a harder problem because execution time depends on complex interactions (such as memory patterns, data dependencies) rather than function invocation counts. 

Fig. \ref{fig:hotdyn1} shows \textit{Dynamis}' results on predicting the \textit{runtime-dominated hotspot functions}, and Fig. \ref{fig:hottime1} illustrates the percentage of application execution runtime covered by these hotspot functions. In our experiments, we considered any function accounting to more than 5\% of the total application runtime as a `hotspot' function. As observed, \textit{Dynamis}' hotspot prediction accuracies and coverage can broadly be divided into three categories: 

\begin{figure*}[htbp]
\centering\includegraphics[width=0.8\textwidth]{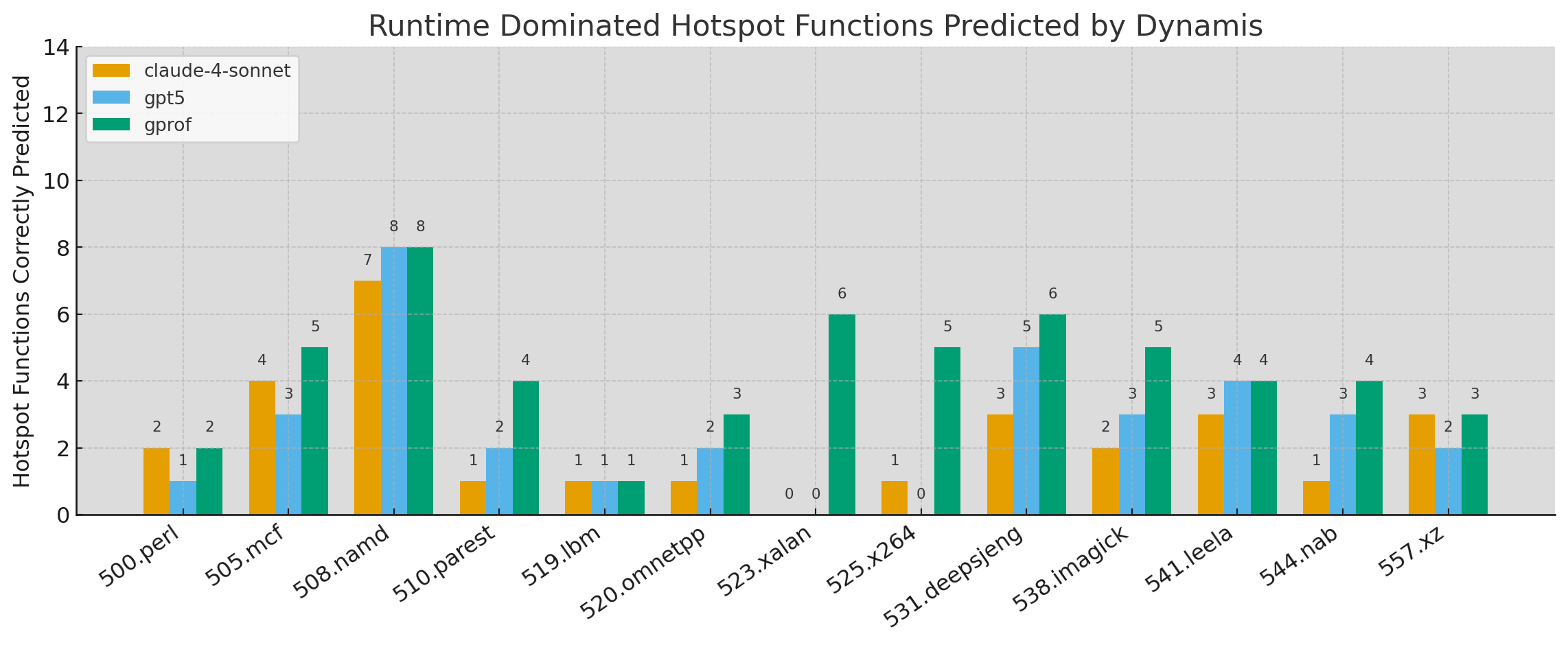}
\caption{\textit{Dynamis}' prediction of hotspot functions in the SPEC2017 suite for REF input. Functions that amount for more than 5\% of the total application execution time are considered here. With the aid of domain knowledge and static compiler analysis, \textit{Dynamis} is able to \textbf{infer the computationally intensive hotspots for large applications without profiling or executing.} }
\label{fig:hotdyn1}
\end{figure*}

$\blacksquare $ \textbf{Near Perfect Coverage} ($> 80\%$ execution predicted): In 6 out of 13 benchmarks from SPEC2017, \textit{GPT5}-equipped \textit{dynamis} was able to extrapolate almost all of the runtime program hotspot functions. These are majorly compute-bound physical/chemical simulation benchmarks, where there exists a direct correlation between domain knowledge and the algorithmic complexity. For instance, in $508.namd\_r$, the LLM first identified the application domain as molecular dynamics, and based on that, it \textit{knew} that ``atoms interact through \textit{long-ranged forces}, which involves atom-pairs", leading to $O (n^2)$ scaling. It was further reinforced with the high inter-procedural loop depth ($\sim 214$) for the aforementioned computation. Similarly, in $519.lbm\_r$, the LLM extrapolated the algorithmic complexity $O(N \cdot T)$, after inferring that the application implements the \textit{Lattice-Boltzmann Method} from fluid mechanics. The rest of the hotspots were then obtained by mapping the call graph to the corresponding functions, and incorporating the loop depths to highlight potential bottlenecks. Also, note that benchmarks from this category range from a few KLOC to $\sim 250$ KLOC.

\begin{figure*}[ht]
\begin{tabular}{p{0.5\textwidth}p{0.5\textwidth}}
    \begin{minipage}{.49\textwidth}
    \centering\includegraphics[width=1.0\textwidth]{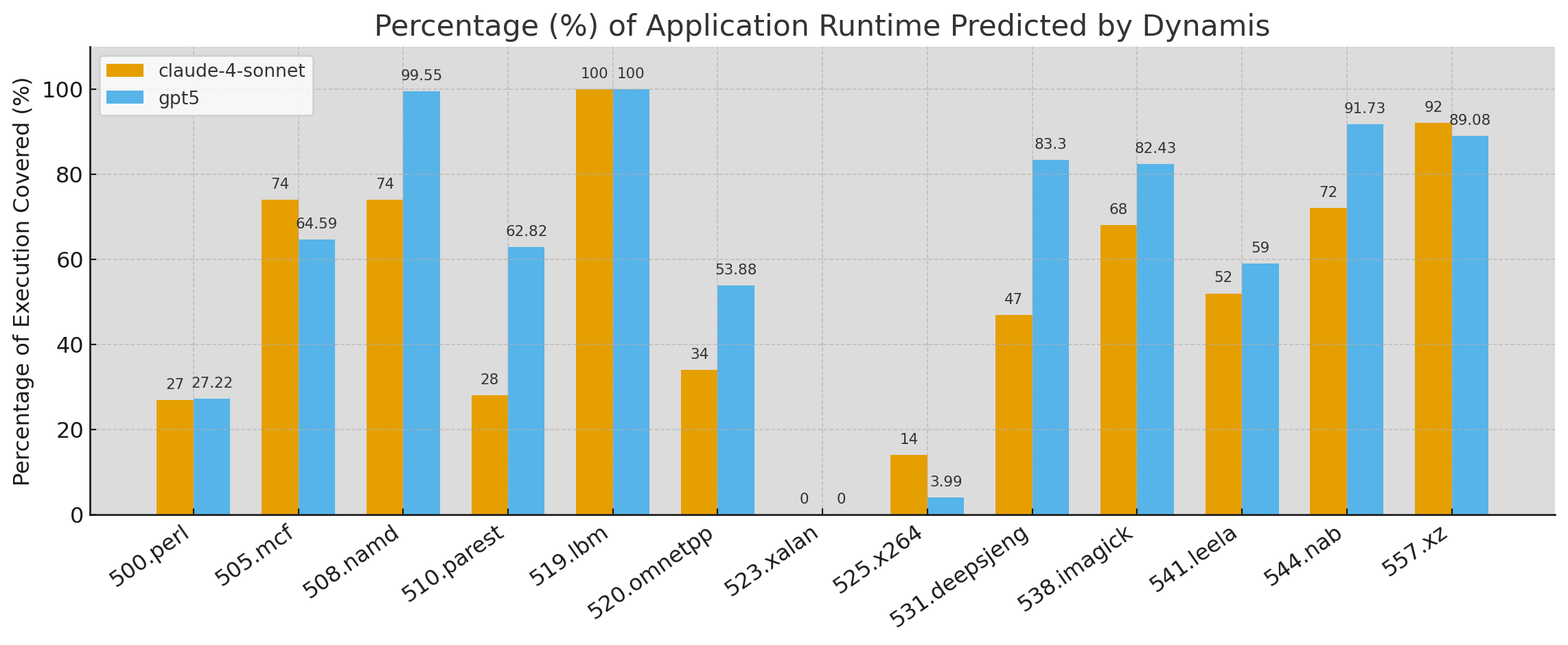}
    \caption{\small Percentage of application execution runtime covered by \textit{Dynamis}'s hotspot predictions for SPEC2017's REF input. \textit{The highest coverage is obtained on compute-intensive applications, while memory-bound applications have lower prediction accuracy}.} 
    \label{fig:hottime1}
    \end{minipage}
    &
    \begin{minipage}{.49\textwidth}
    \centering\includegraphics[width=1.0\textwidth]{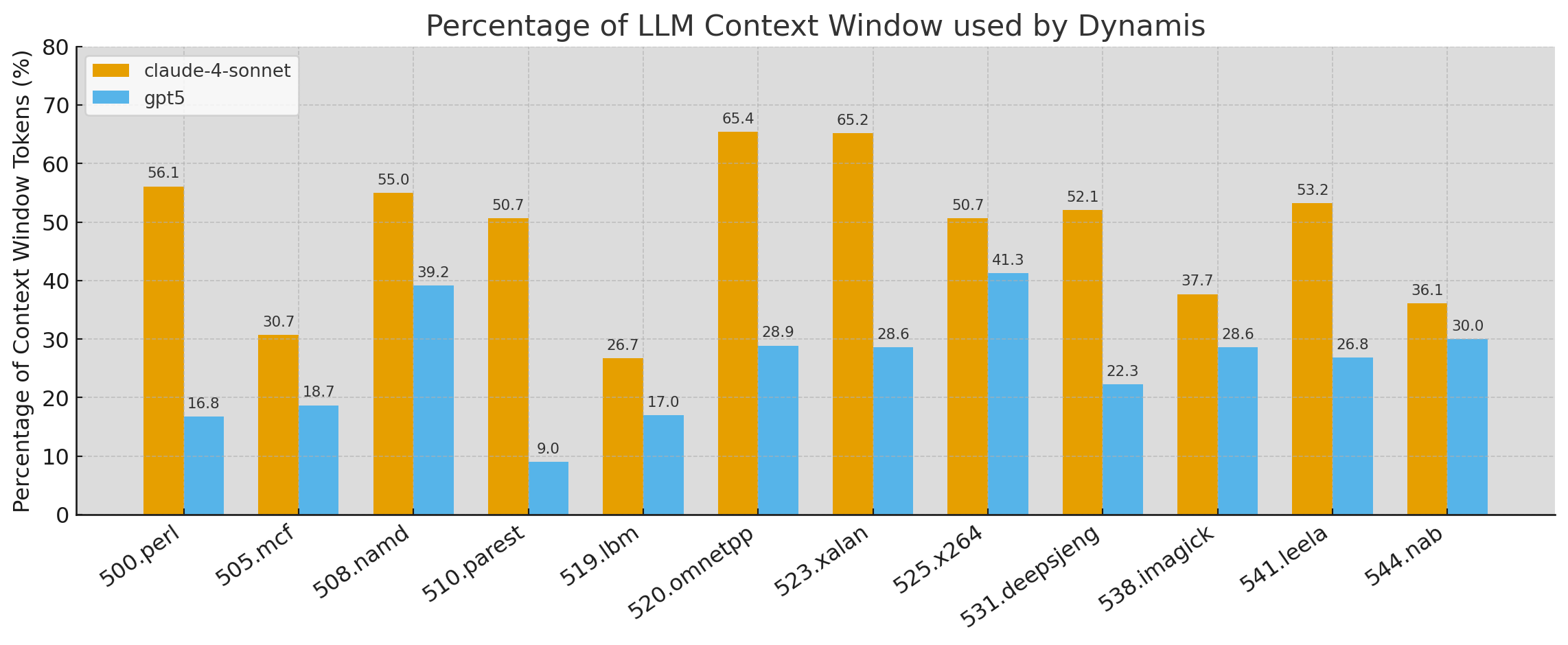}
    \caption{\small Percentage of context window used by \textit{Dynamis} for predicting hotspot functions. A state of the art LLM such as GPT5 \textbf{barely requires a third of its context window to predict input-specific dynamic behavior.}} 
    \label{fig:context1}
    \end{minipage}
\end{tabular}
\end{figure*}

$\blacksquare $ \textbf{Moderate Coverage} ($50 - 80\%$): 4 benchmarks ($500.mcf\_r$, $510.parest\_r$, $520.omnetpp\_r$, $541.leela\_r$) exhibited an execution coverage of 60\% on average. These benchmarks are applications where the hotspot functions are not confined to the primary algorithmic kernels but also include auxiliary and system-level routines that dominate at scale. For instance, in $505.mcf\_r$, the domain knowledge combined with the semantic program artifacts correctly identified 4/5 hotspot functions by zeroing in on the core algorithmic kernels. However, both the LLMs missed the utility sorting function $spec\_qsort$, which dominates at scale. Similarly, in $510.parest\_r$, the LLM was able to identify algorithmic hotpots, but failed to reason about the various C++ template variants missing a couple of hotspots. One possible solution to improve coverage for such benchmarks would be to incorporate finer-grained static analyses—such as resolving function pointers and template instantiations.   

$\blacksquare $ \textbf{Minimal Coverage} ($0 - 35\%$): The runtime hotspots in $523.xalanc\_r$, are dominated by memory allocation/deallocation patterns, and object lifecycle management. This makes the application memory-bound, and the LLM's domain expertise, which focused on determining algorithmic complexity, did not translate to actual hotspots. On the other hand, $500.perlbench\_r$, exhibits a near-uniform distribution of execution time ($\sim 1-2\%$ each), with two outlier functions amounting to only a third of the runtime. \textit{Dynamis} was able to infer the dominant outliers ($S\_regmatch$, $perl\_hv\_common$), leading to low coverage. $525.x264\_r$ has a similar distribution, with the longest function only lasting 4 secs. 

\subsection{Leveraging Phaedrus: Profile-Guided Optimization}
\label{res:pgo}

$\star$ \textbf{Runtime-Dominated Hotspot Functions}: We also used \textit{Dynamis}' prediction of hotspot functions (\S \ref{res:hotpotpred}) to perform PGO. The performance improvement over a normalized O3-baseline is shown in Fig. \ref{fig:hotspottime2}. \textit{Dynamis} achieves a performance improvement of \textit{6}\% and \textit{3.63\%} over the traditional PGO baseline with claude-4-sonnet and GPT-5 respectively. An important point to note is that even though \textit{Dynamis}’ predictions covered only 50–60\% (Fig.\ref{fig:hotdyn1}) of the total execution time (and less than 30\% in some cases), its inferred hotspots were still broadly consistent with those observed in the ground-truth \textit{profdata} profiles. This is particularly interesting because it suggests that \textit{even with runtime partial coverage, the framework captured the core computational patterns responsible for most of the runtime improvement}.

To investigate this effect further, we analyzed the contents of the predicted \textit{profdata} and the ground truth closely. Profiling tool \textit{gprof} highlights expensive compute-bound routines that dominate wall-clock time, while \texttt{profdata} emphasizes fast but frequently executed functions. It became evident that functions dominating runtime are not always those most amenable to direct optimization. In some cases, performance can be improved indirectly—by transforming other functions that invoke these heavy kernels. For instance, in $520.omnetpp\_r$, the low-level function \textit{cMessageHeap::shiftup} predicted by \textit{Dyanmis} appeared trivial in \textit{gprof}, yet it dominated execution counts in \texttt{profdata}. We observed that \textit{Dynamis}’ reasoning consistently balanced between capturing both computationally intensive and optimization-relevant routines. \textit{In spite of being prompted to identify runtime-dominated hotspots, \textit{Dynamis} still surfaced functions that were most likely to yield measurable performance improvement when optimized.} This effect was also observed in the application binary sizes, where \textit{Dyanmis} achieved an improvement of \textit{5.19} over O3. This is achieved by transforming the function layout within the application binary. This optimization reduces the binary size by enabling superior function compression \cite{hoag2024reordering}, and can potentially improve performance as the new layout would minimize page faults by appropriately ordering them. 

\begin{figure*}[ht]
\begin{tabular}{p{0.5\textwidth}p{0.5\textwidth}}
    \begin{minipage}{.49\textwidth}
    \centering\includegraphics[width=1.0\textwidth]{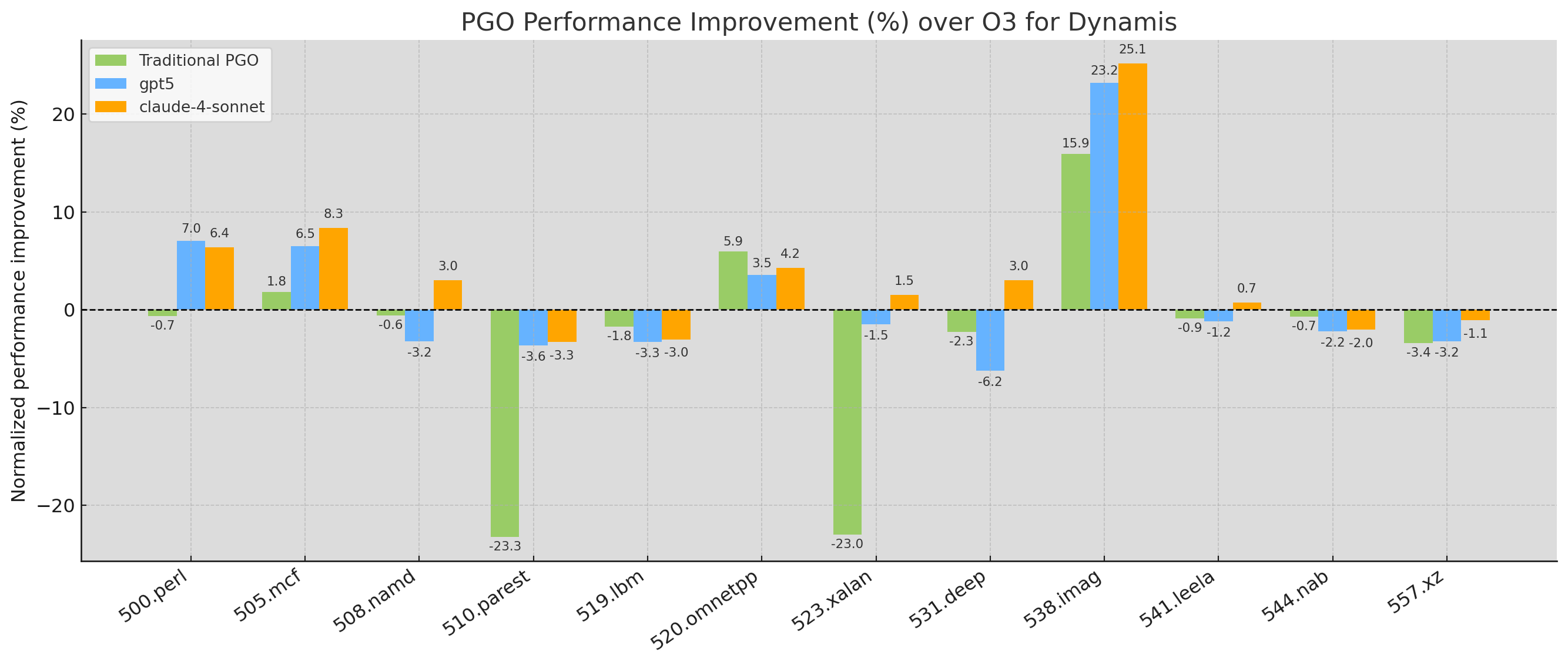}
    \caption{\small \textit{Dynamis}' hotspot predictions achieves an average performance improvement of 6.08\% over traditional PGO, without any profiling. This is achieved by capturing both computationally intensive and optimization-relevant routines by \textit{Dynamis}. The time shown is normalized with respect to O3.} 
    \label{fig:hotspottime2}
    \end{minipage}
    &
    \begin{minipage}{.49\textwidth}
    \centering\includegraphics[width=1.0\textwidth]{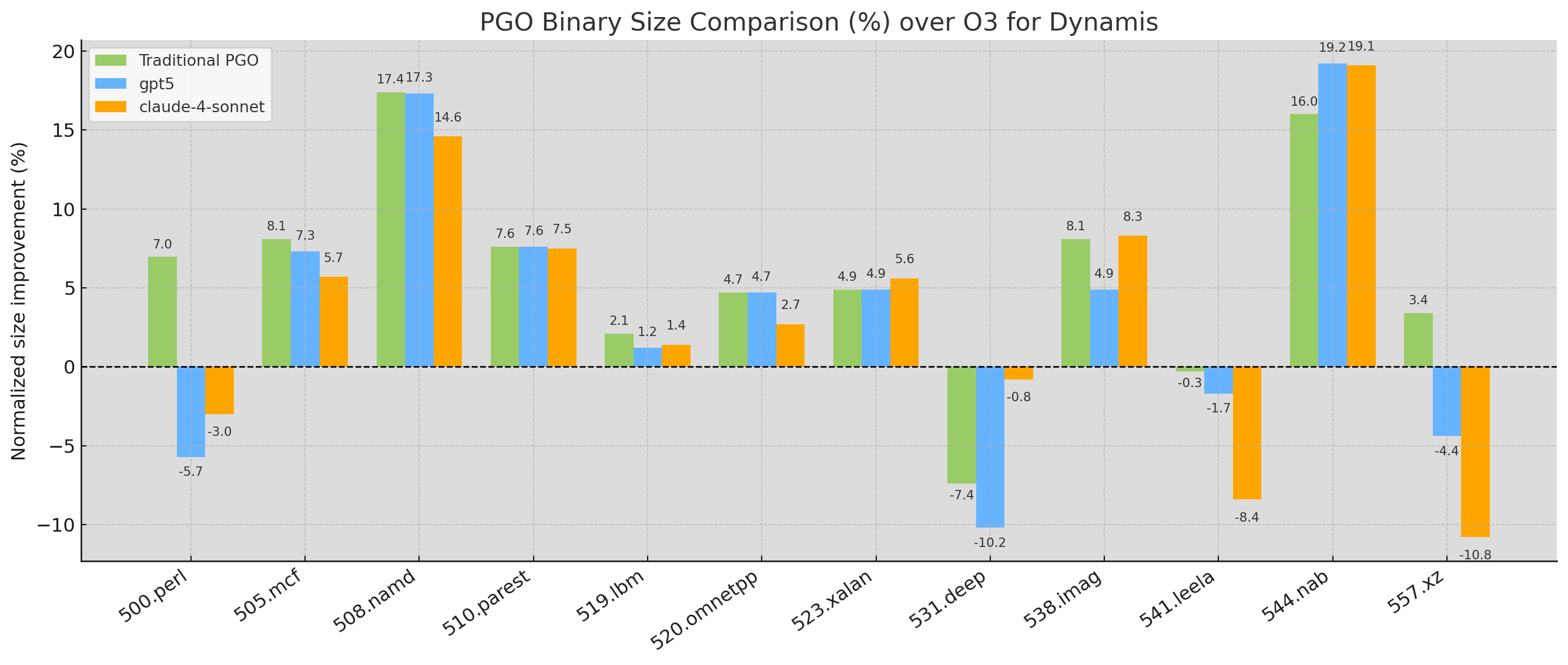}
    \caption{\small \textit{Dyanmis} hotspot predictions achieves average codesize reduction of 4.53\% \& 3.26\% over O3 with claude-4-sonnet and GPT-5 respectively. This is achieved by function re-layout transformation which enables superior function compression. The size shown is normalized with respect to O3.} 
    \label{fig:bin1}
    \end{minipage}
\end{tabular}
\end{figure*}

$\star$ \textbf{Frequently Executed Functions}: We also leveraged \textit{Phaedrus}'s outputs (both \textit{Morpheus} \& \textit{Dynamis}) of \textit{most frequently executed functions} to perform PGO. Fig. \ref{fig:perfmorph} \& Fig. \ref{fig:perfllm} shows the effect of performing PGO with \textit{Phaedrus}' outputs. As observed, we notice performance improvements comparable to traditional PGO (0.8\% on average and upto 19.18\% at max) on most benchmarks. 

On taking a closer look into the performance trends, we find that on average \textit{Morpheus} outperforms \textit{Dynamis}' frequently executed functions setting by 2\% on average for both baselines. This is explained by the fact that \textit{Morpheus}'s RNN model generates the unified trace in a timestep manner; rather than extrapolating the computationally hot functions, as \textit{Dynamis} does. This highlights a fundamental difference in the methodology of both approaches: \textit{application behavior synthesis, a profile-less approach, is approximate than application profile generalization, which involves profile generation, compression and unification}. This intuition is clearly reflected in Fig. \ref{fig:perfmorph} \& \ref{fig:perfllm}. 

Fig. \ref{fig:pgoresdyn} highlights the binary size reduction for \textit{Dynamis} over both baseline. \textit{Dynamis} outperforms the normal O3 baseline (no PGO) by 14\% on average with up to 65\% in $519.lbm\_r$. As mentioned earlier, the primary reason for size reduction is the high compression obtained by appropriate function layout by LLVM. Fig \ref{fig:pgoresmorp} depicts the binary size reduction over both baseline - in this case, we don't notice any observable difference.  

\begin{figure*}[ht]
\begin{tabular}{p{0.5\textwidth}p{0.5\textwidth}}
    \begin{minipage}{.49\textwidth}
        \centering\includegraphics[width=1.0\textwidth]{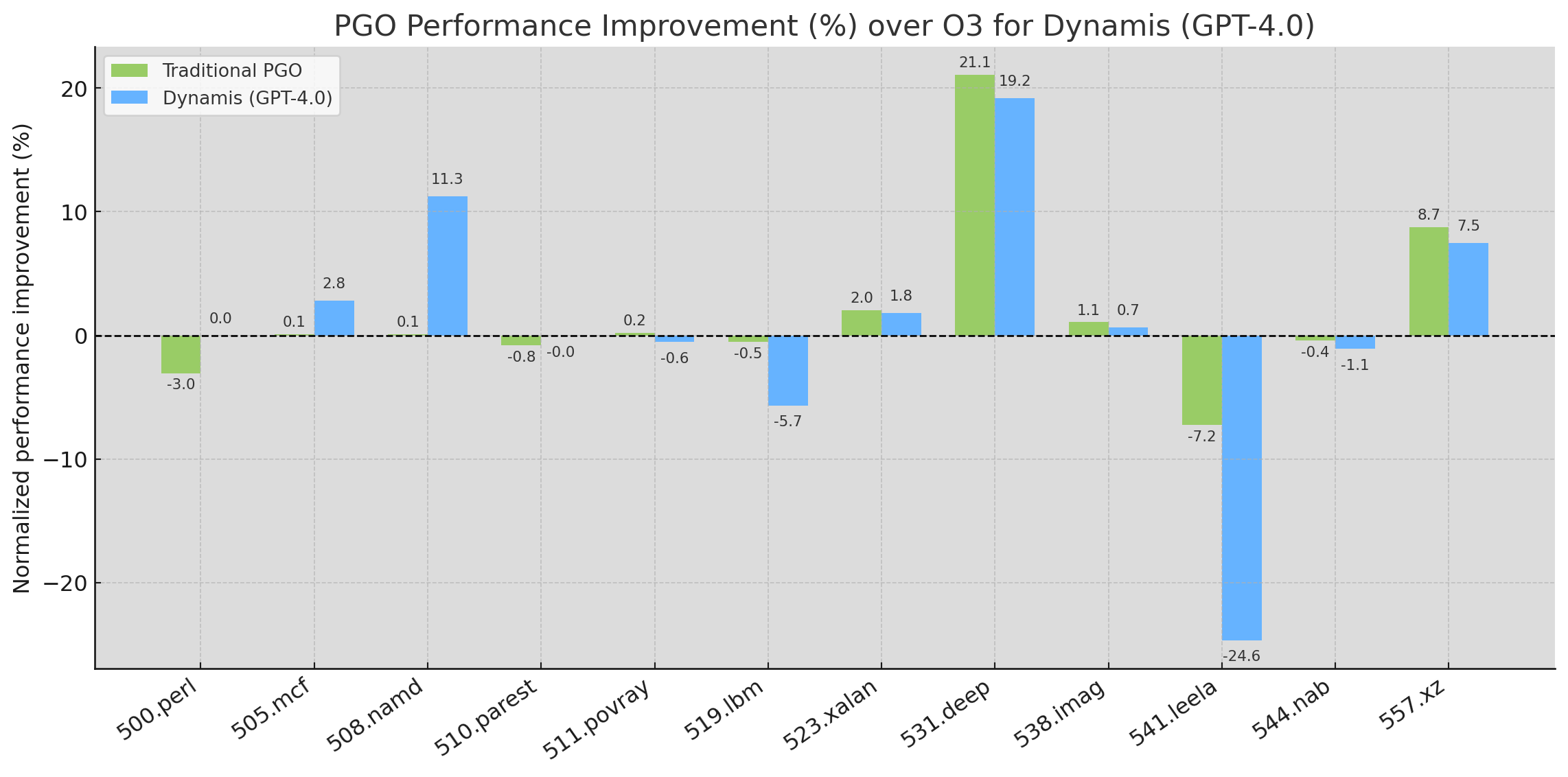}
    \caption{\small \textit{Dynamis}' frequently executed predictions achieves around 0.8\% average improvement over O3. Being a legacy model, GPT 4.0 often missed the exact computational hotspots in large benchmarks. } 
    \label{fig:perfllm}
    \end{minipage}
    &
    \begin{minipage}{.49\textwidth}
    
    \centering\includegraphics[width=1.0\textwidth]{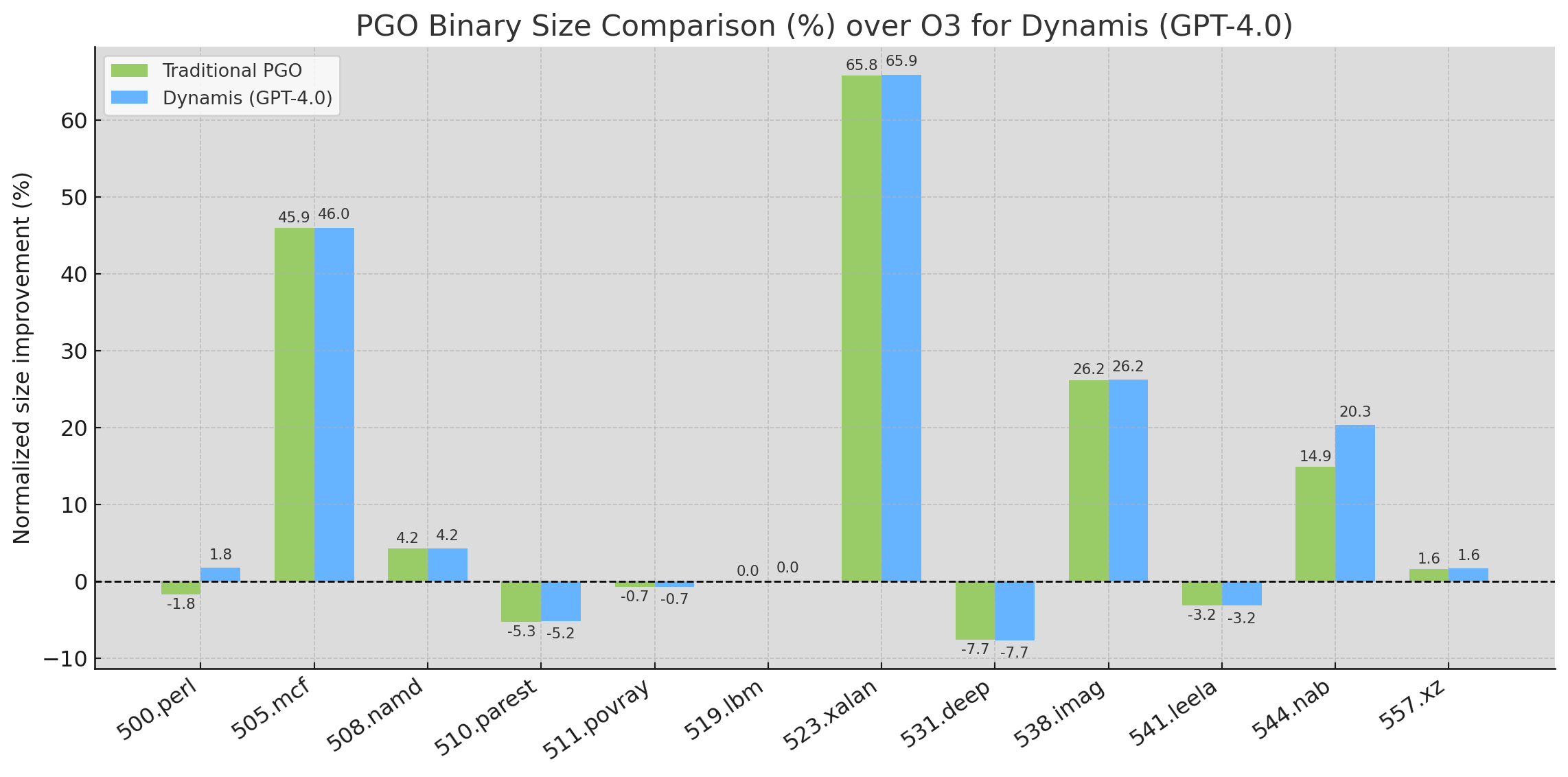}
    \caption{\small \textit{Dynamis}' frequently executed predictions reduced codesize by an average of 13.68\% over O3, outperfoming the hotspot predictions.} 
    \label{fig:pgoresdyn}
    \end{minipage}
\end{tabular}
\end{figure*}

\begin{figure*}[htbp]
\begin{tabular}{p{0.5\textwidth}p{0.5\textwidth}}
    \begin{minipage}{.49\textwidth}
    \centering\includegraphics[width=1.0\textwidth]{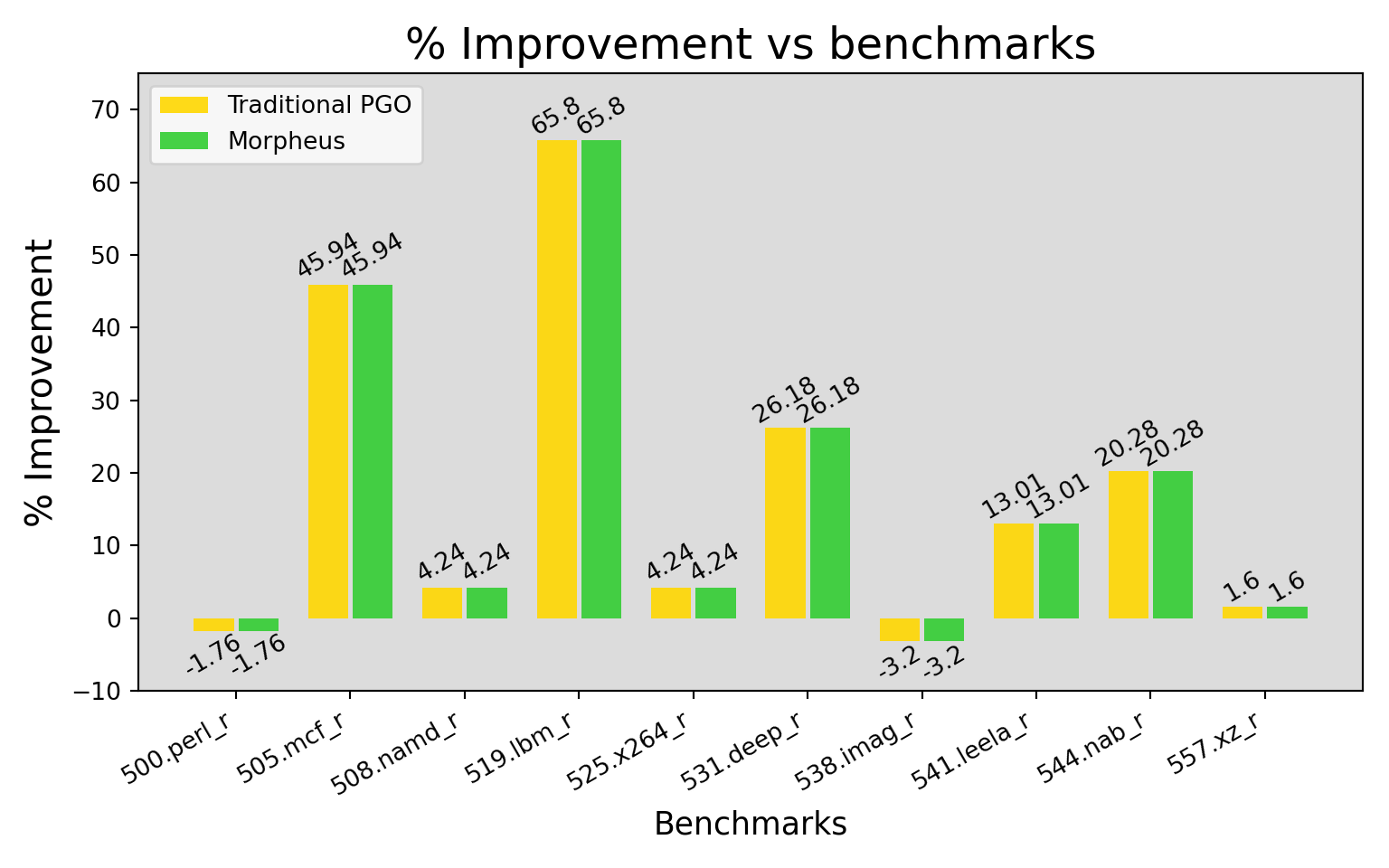}
    \caption{\small \textit{Morpheus}'s results for codesize reduction. An average of 17.63\% improvement is observed, just as good as traditional PGO} 
    \label{fig:pgoresmorp}
    \end{minipage}
    &
    \begin{minipage}{.49\textwidth}
    \centering\includegraphics[width=1.0\textwidth]{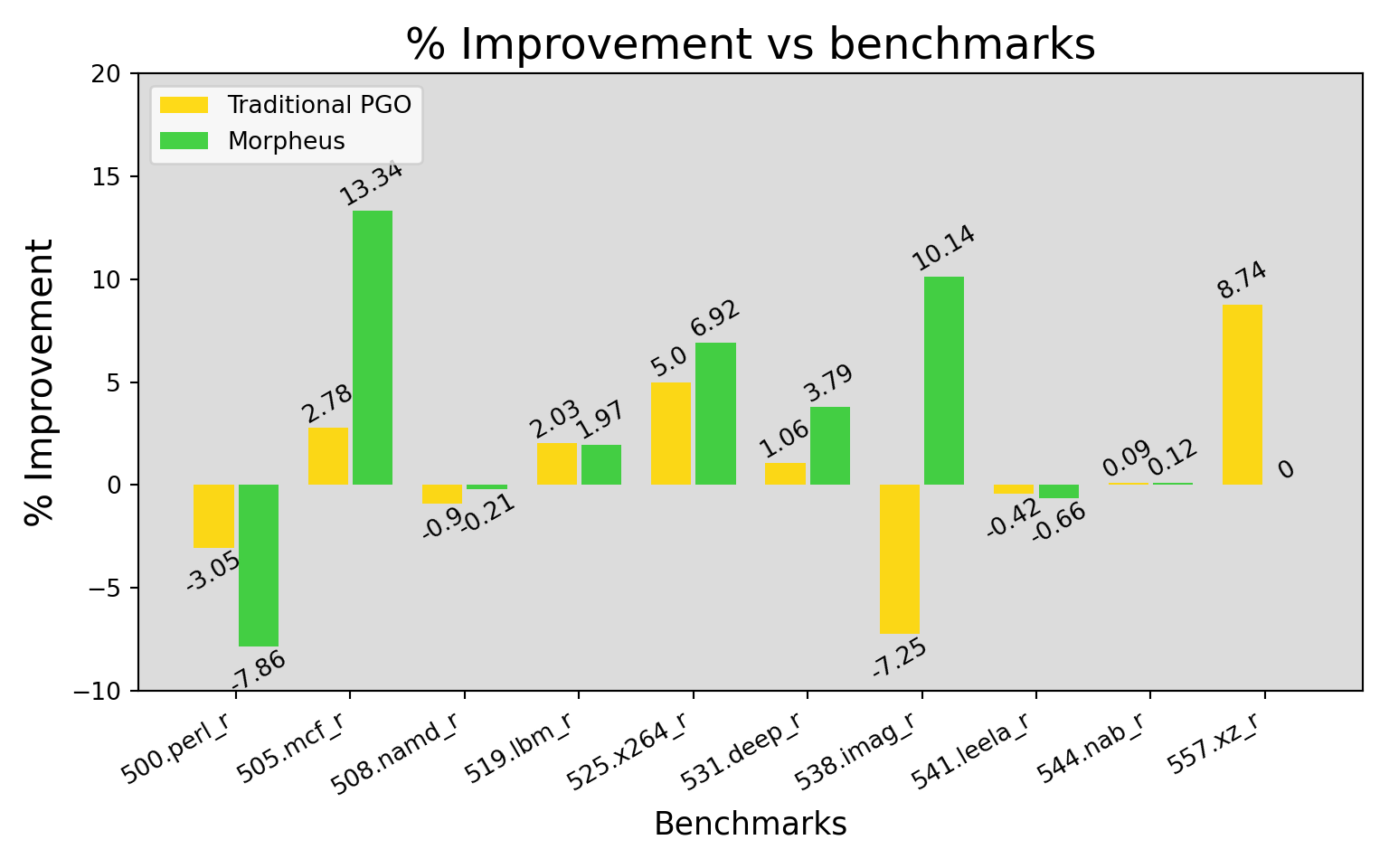}
    \caption{\small \textit{Morpheus}'s performance improvements observed across different benchmarks. Baseline here is normal compilation without profiling. \textit{Morpheus} achieves around 2.8\% average improvement over the baseline.} 
    \label{fig:perfmorph}
    \end{minipage}
\end{tabular}
\end{figure*}

\subsection{Case Study: Real-World Benchmarks}

In order to better demonstrate the capabilities of \textit{Dynamis} on larger real world programs, we evaluated the framework on 8 popular graph processing problems (\cite{beamer2017gapbenchmarksuite}) (BFS, CC, CC\_SV, TC, PR, PR\_SPMV, SSSP) with an undirected graph of degree 15 with 16777216 nodes and 268435177 undirected edges, and also with the popular compiler benchmark \textit{GCC}, which has 1304 KLOC. Fig. \ref{fig:hotspottime3} shows that \textit{Dynamis} outperforms O3 baseline by 4.45\% for GCC. It navigated the huge codebase by starting with the program inputs \& flags, based on which it inferred the dominant compiler passes. Then it identified the "known" heavy kernels within those passes (bitmaps and graph operations, points-to analysis, etc), which were then ranked by their expected algorithmic cost. This shows how \textit{Dynamis} is able to adapt its reasoning for large codebase by tracing its thinking from the inputs.

\begin{figure*}[ht]
\begin{tabular}{p{0.5\textwidth}p{0.5\textwidth}}
    \begin{minipage}{.49\textwidth}
    \centering\includegraphics[width=1.0\textwidth]{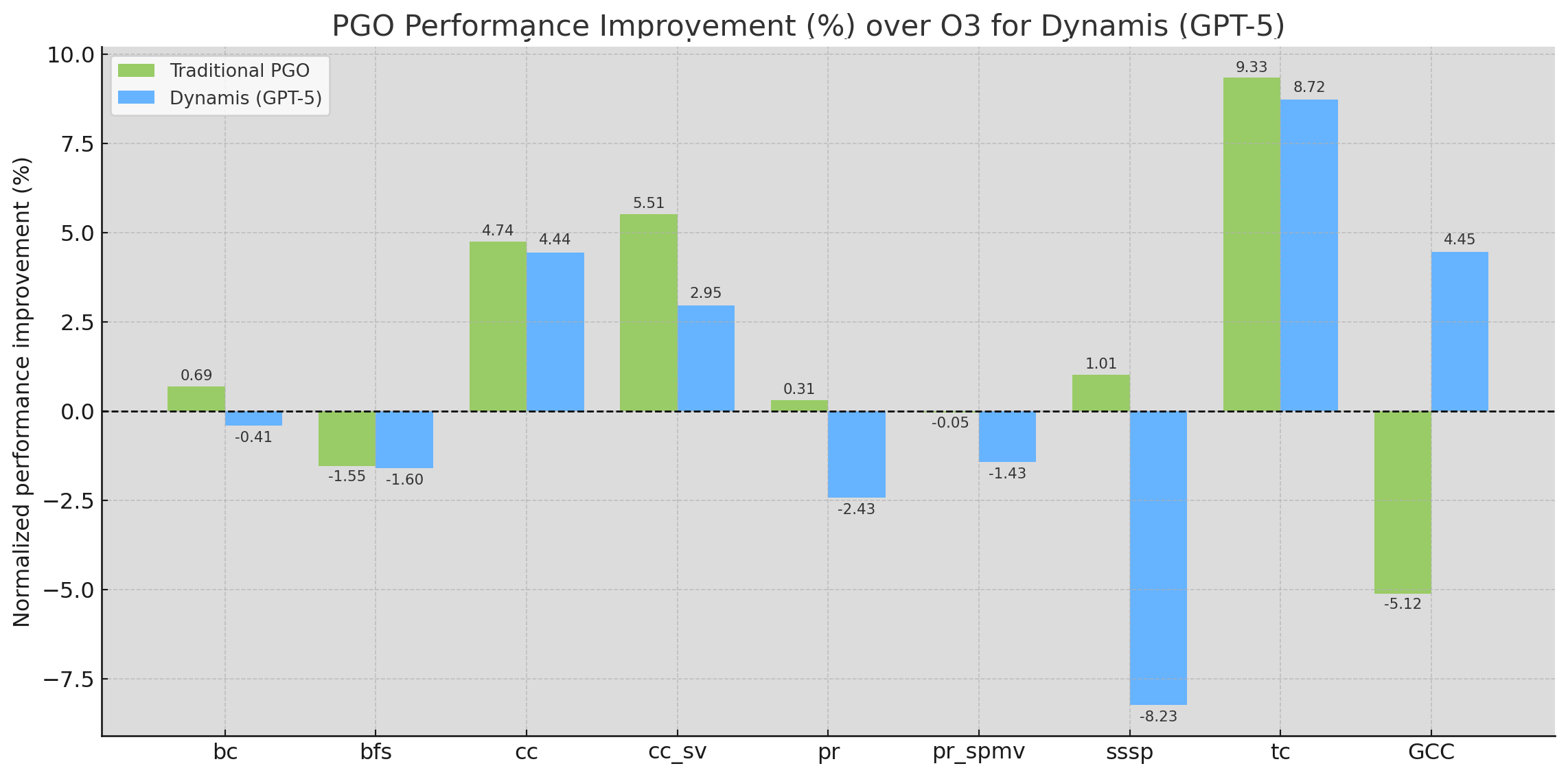}
    \caption{\small \textit{Dynamis}' hotspot predictions on real-world programs showed that it can achieves a performance improvement of 4.45\% and 9.09\% over O3 and traditional PGO in GCC. The average improvement is -0.92\% over O3, which is a slight slowdown.} 
    \label{fig:hotspottime3}
    \end{minipage}
    &
    \begin{minipage}{.49\textwidth}
    \centering\includegraphics[width=1.0\textwidth]{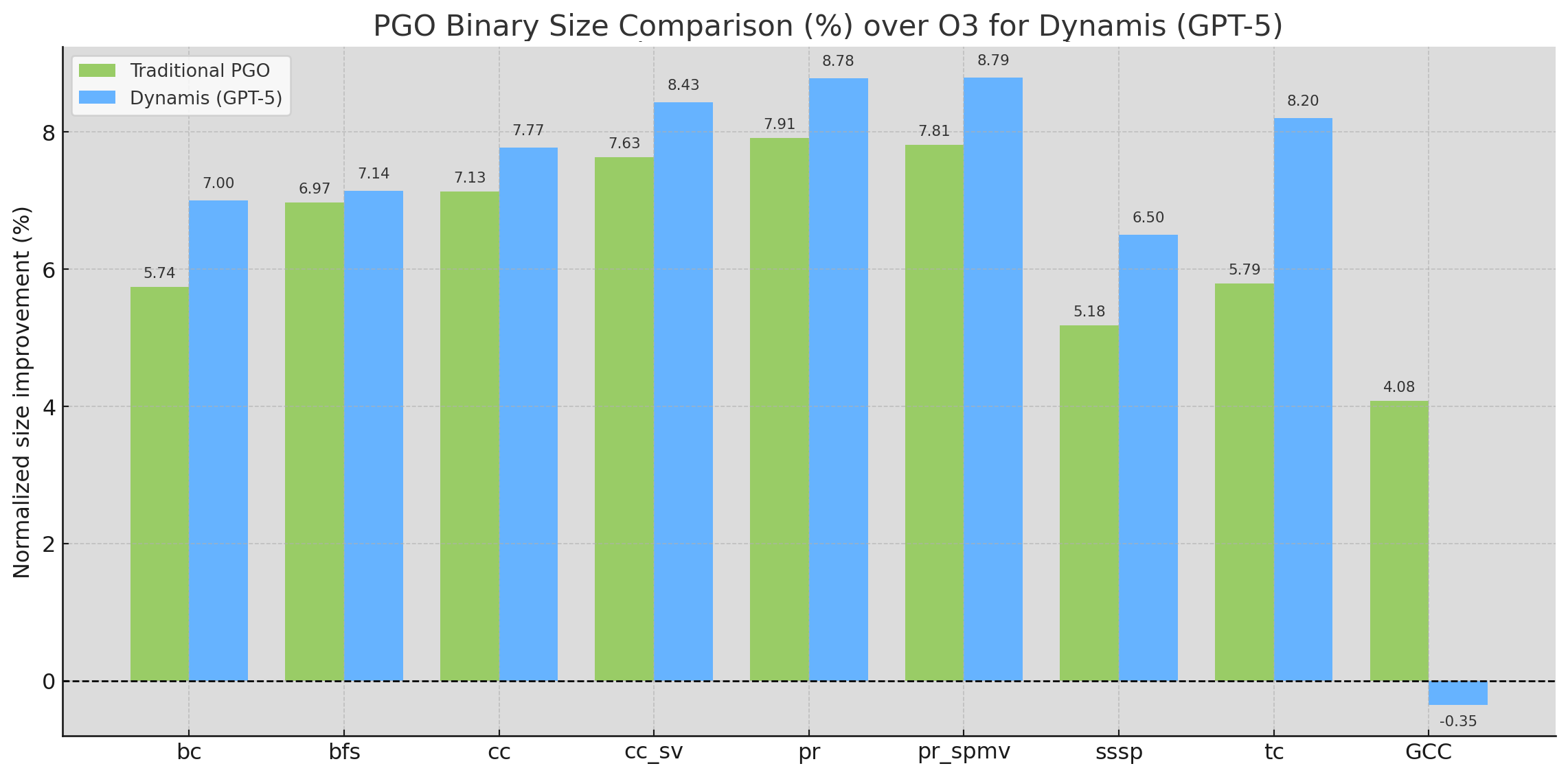}
    \caption{\small \textit{Dynamis}' hotspot predictions on real-world showed an average improvement of 6.95\% over O3, and is slightly better than traditional PGO on average. The graph benchmarks showed the most reduction in size.} 
    \label{fig:bin4}
    \end{minipage}
\end{tabular}
\end{figure*}

On the other hand, in the graph workloads, \textit{Dynamis} switched back its default strategy of first leveraging domain knowledge (such as union-find and path compression in CC) to identify compute-heavy kernels, and then mapping this theoretical knowledge to the source codes. It also leveraged the call graph to filter low-level STL calls, then determined a ranking based on the given input. In this case, \textit{Dynamis} was able to achieve an average reduction in the application codesize of 7.83\% over O3, but it did not translate to significant performance improvement ($0.37\%$ over O3). A possible reason for this is that these programs BFS/PR/SSSP/CC are memory-bound and spend time on random memory access; reducing code size doesn’t impact bandwidth/latency.

\noindent

\subsection{Evaluating Morpheus: Application Profile Generalization}
\label{res:morph}

$\star$ \textbf{Need for Compressing WPP based Function Profiles}: Table \ref{tab1} shows the total tokens and disk size for WPP based function profiles expressed in terms of application function calls for SPEC2k17 \cite{bucek2018spec} benchmarks. As shown, the total function call tokens can range up to $10^{10}$ tokens per profile while incurring disk sizes of a few GB to $10^2$ GB. Without the compression techniques employed by \textit{Morpheus}, such profiles cannot be analyzed. For instance, Fig. \ref{fig:motiv1a} summarizes the result of training a simple RNN model with one hidden layer on the uncompressed WPP data generated from the smallest program input, in SPEC2k17 Benchmarks. Out of 14 benchmarks, the training process on 7 benchmarks failed because of exceeding the tensor's default memory capacity in the \textit{Pytorch} framework, while 5 benchmarks had \textit{out-of-memory} errors, and were terminated by the OS scheduler. Furthermore, Fig. \ref{fig:motiv1b} depicts the low ratio of unique tokens in WPP based function profiles. This indicates a significantly high rate of duplicate (and redundant) tokens, which can lead to excessive bias and unnecessary computations during training.  

\begin{figure*}[!ht]
\begin{tabular}{p{0.51\textwidth}p{0.49\textwidth}}
    \begin{minipage}{.50\textwidth}
    \footnotesize
    \centering
    \resizebox{\columnwidth}{!}{
    \begin{tabular}{|l|lll|lll|}
\hline
\multirow{2}{*}{\textbf{Benchmark}} & \multicolumn{3}{c|}{\textbf{WPP based function profile Disk Size}} & \multicolumn{3}{c|}{\textbf{Function Tokens}} \\ \cline{2-7} 
 & \multicolumn{1}{l|}{\textbf{Small}} & \multicolumn{1}{l|}{\textbf{Medium}} & \textbf{Large} & \multicolumn{1}{l|}{\textbf{Small}} & \multicolumn{1}{l|}{\textbf{Medium}} & \textbf{Large} \\ \hline
500.perl\_r & \multicolumn{1}{l|}{305 MB} & \multicolumn{1}{l|}{887 MB} & 8.2 GB & \multicolumn{1}{l|}{66.5 M} & \multicolumn{1}{l|}{198.2 M} & 1.8 B \\ \hline
502.gcc\_r & \multicolumn{1}{l|}{1.3 MB} & \multicolumn{1}{l|}{11 GB} & {19 GB} & \multicolumn{1}{l|}{248 K} & \multicolumn{1}{l|}{2.2 B} & {4.3 B} \\ \hline
505.mcf\_r & \multicolumn{1}{l|}{2 GB} & \multicolumn{1}{l|}{11 GB} & 58 GB & \multicolumn{1}{l|}{687 M} & \multicolumn{1}{l|}{3.5 B} & 20.5 B \\ \hline
508.namd\_r & \multicolumn{1}{l|}{25 MB} & \multicolumn{1}{l|}{73 MB} & 1.6 GB & \multicolumn{1}{l|}{8.6 M} & \multicolumn{1}{l|}{60.1 M} & 558 M \\ \hline
510.parest\_r & \multicolumn{1}{l|}{110 MB} & \multicolumn{1}{l|}{692 MB} & 38 GB & \multicolumn{1}{l|}{6 M} & \multicolumn{1}{l|}{35 M} & 2.2 B \\ \hline
519.lbm\_r & \multicolumn{1}{l|}{4 KB} & \multicolumn{1}{l|}{4 KB} & 16 KB & \multicolumn{1}{l|}{75} & \multicolumn{1}{l|}{919} & 6056 \\ \hline
520.omnetpp\_r & \multicolumn{1}{l|}{853 MB} & \multicolumn{1}{l|}{14 GB} & 124 GB & \multicolumn{1}{l|}{4.4 M} & \multicolumn{1}{l|}{43 M} & {378 M} \\ \hline
523.xalanc\_r & \multicolumn{1}{l|}{2.2 MB} & \multicolumn{1}{l|}{854 MB} & 95 GB & \multicolumn{1}{l|}{421 K} & \multicolumn{1}{l|}{127 M} & {18.9 B} \\ \hline
525.x264\_r & \multicolumn{1}{l|}{25 MB} & \multicolumn{1}{l|}{25 MB} & 387 MB & \multicolumn{1}{l|}{3.8 M} & \multicolumn{1}{l|}{3.8 M} & 40 M \\ \hline
526.blender\_r & \multicolumn{1}{l|}{56 MB} & \multicolumn{1}{l|}{8.8 GB} & 27 GB & \multicolumn{1}{l|}{10.6 M} & \multicolumn{1}{l|}{1.8 B} & {5.3 B} \\ \hline
531.deep\_r & \multicolumn{1}{l|}{2.4 GB} & \multicolumn{1}{l|}{20 GB} & 110 GB & \multicolumn{1}{l|}{510 M} & \multicolumn{1}{l|}{5.1 B} & {39 B} \\ \hline
538.imagick\_r & \multicolumn{1}{l|}{248 KB} & \multicolumn{1}{l|}{183 MB} & 90 GB & \multicolumn{1}{l|}{55 K} & \multicolumn{1}{l|}{47 M} & {24 B} \\ \hline
541.leela\_r & \multicolumn{1}{l|}{2.3 GB} & \multicolumn{1}{l|}{38 GB} & 169 GB & \multicolumn{1}{l|}{454 M} & \multicolumn{1}{l|}{7.9 B} & {58 B} \\ \hline
544.nab\_r & \multicolumn{1}{l|}{4.7 MB} & \multicolumn{1}{l|}{25 MB} & 2.6 GB & \multicolumn{1}{l|}{76 K} & \multicolumn{1}{l|}{185 K} & 1.6 M \\ \hline
557.xz\_r & \multicolumn{1}{l|}{252 KB} & \multicolumn{1}{l|}{297 MB} & 2 GB & \multicolumn{1}{l|}{11.8 K} & \multicolumn{1}{l|}{76 M} & 537 M \\ \hline
\end{tabular}}
\caption{\small Total Tokens and their corresponding Disk Sizes, for WPP based function profiles consisting of function calls in SPEC2k17 (K = $10^3$, M = $10^6$, B = $10^9$)}
\label{tab1}
    \end{minipage}
    &
    \begin{minipage}{.48\textwidth}
    \centering\includegraphics[width=0.9\textwidth]{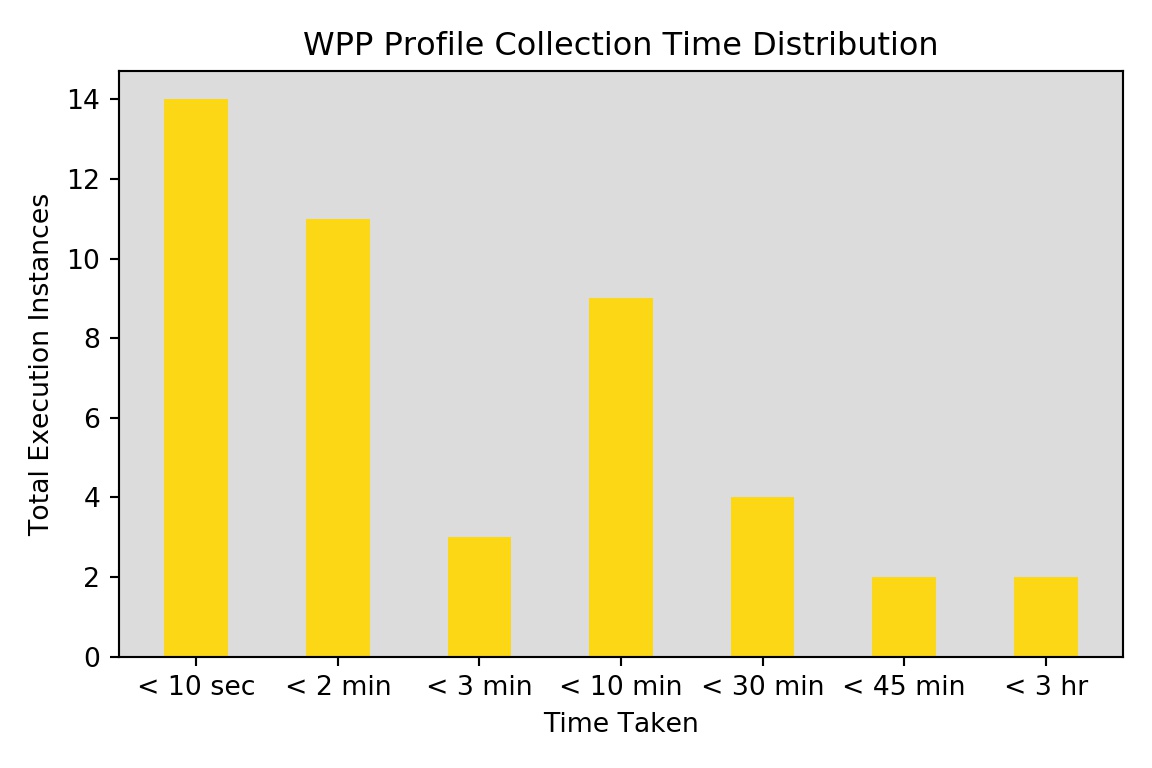}
    %\vspace{-15pt}
\caption{\small Time Taken for Collecting WPP based function profiles for different Input Sizes (small, medium, large) in SPEC Suite. In most cases, WPP based function profile on a convenient input can be generated within 10 secs.} 
    \label{fig:time1}
\end{minipage}
\end{tabular}
\end{figure*}

\begin{figure*}[ht]
\begin{tabular}{p{0.5\textwidth}p{0.5\textwidth}}
    \begin{minipage}{.49\textwidth}
    \centering\includegraphics[width=1.0\textwidth]{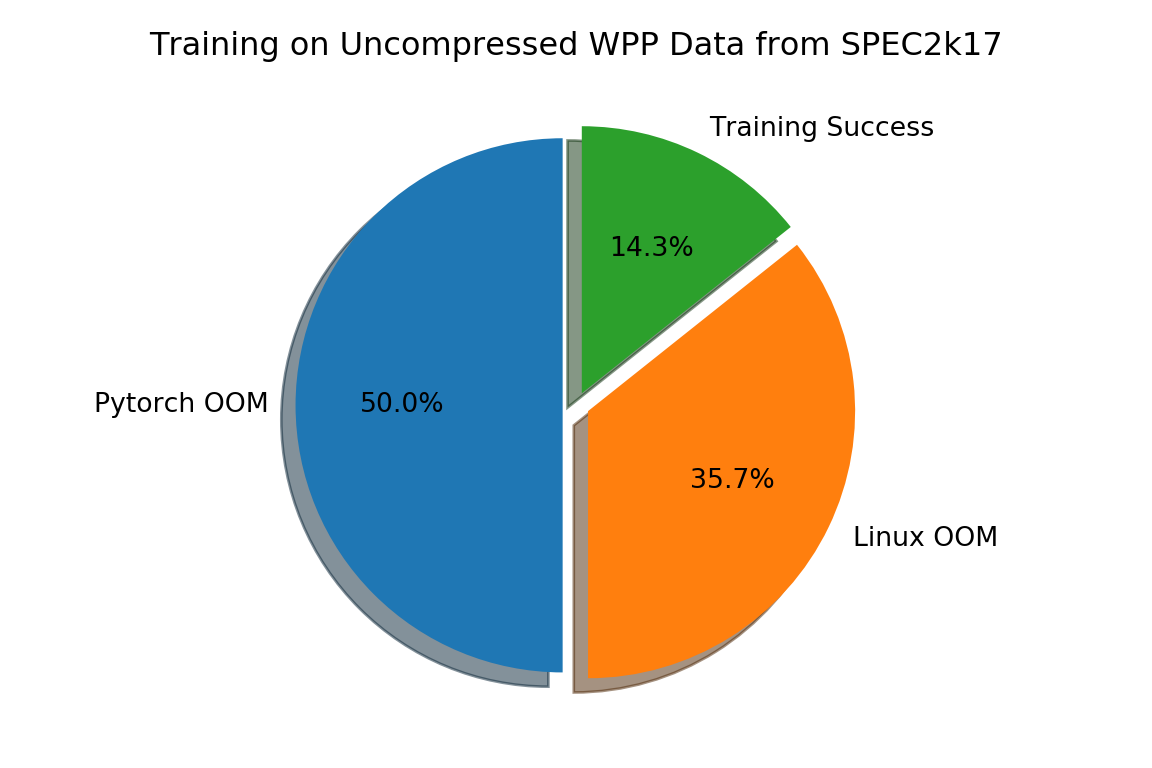}
\caption{Results of training RNN models directly on the uncompressed WPP data in SPEC2k17 suite. The training was successful only on 2 benchmarks, while the rest (13 benchmarks) led to out of memory and memory allocation errors.} 
\label{fig:motiv1a}
    \end{minipage}
    &
    \begin{minipage}{.49\textwidth}
    \centering\includegraphics[width=1.0\textwidth]{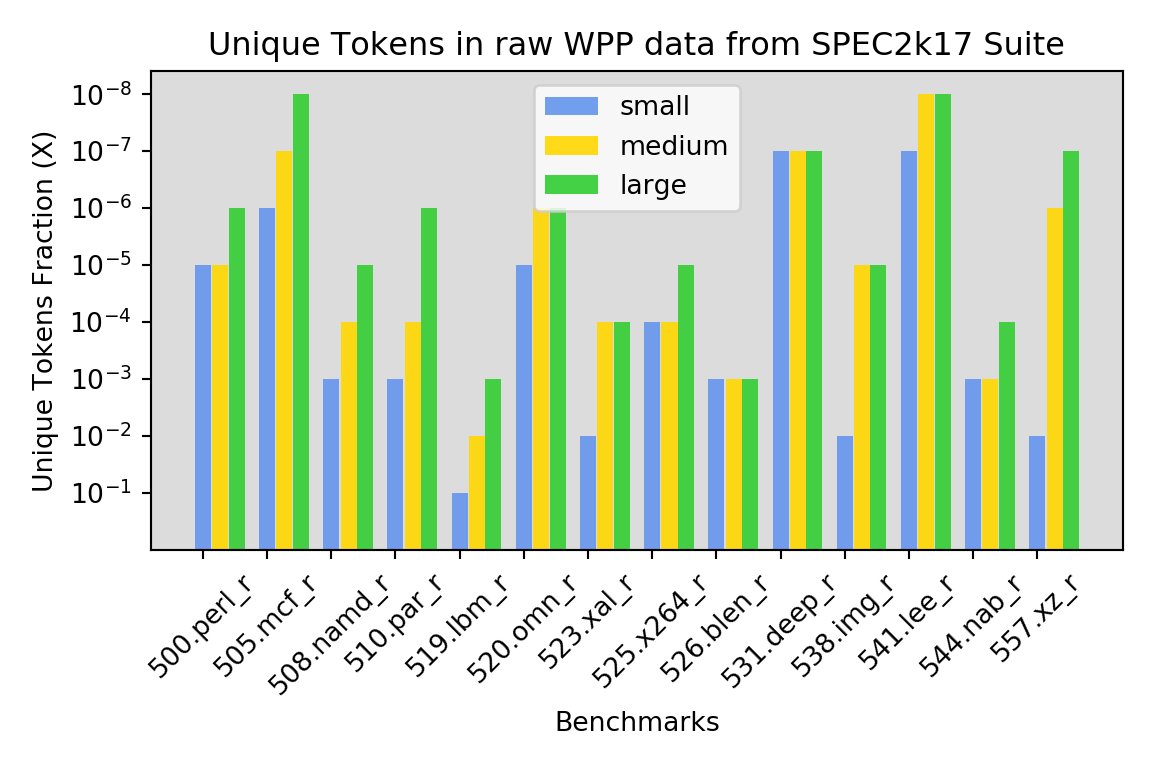}
    \vspace{-15pt}
\caption{Fraction of unique tokens in WPP data on all inputs in SPEC2k17. The plot indicates a very low fraction of unique tokens in the WPP based function profile, which is also inversely proportional to the input size. } 
    \label{fig:motiv1b}
\end{minipage}
\end{tabular}
\end{figure*}

$\star$ \textbf{Smart-Loop Instrumentation}:  Table \ref{res:tab1} depicts the size of WPP based function profiles collected after the smart-loop instrumentation scheme. As observed, the original WPP based function profiles that were originally in GBs, were reduced to mere KBs in most cases. Fig. \ref{fig:res2} shows the ratio of reduction achieved in WPP based function profiles for all input sizes. For certain large profiles (\textit{531.deepsjeng\_r, 538.imagick\_r, 541.leela\_r}), the compression ratio is upto $10^7\times$, which points to the high degree of redundant tokens. This makes the WPP based function profiles scalable, and makes the way for further analysis for training generative models.

\begin{figure*}[!ht]
\begin{tabular}{p{0.48\textwidth}p{0.52\textwidth}}
    \begin{minipage}{.47\textwidth}
    \footnotesize
    \centering
    \begin{tabular}{|l|lll|}
\hline
\multicolumn{1}{|c|}{\multirow{2}{*}{\textbf{Benchmark}}} & \multicolumn{3}{c|}{\textbf{\begin{tabular}[c]{@{}c@{}}Smart Loop Instrumented \\ WPP based function profile Disk Size\end{tabular}}} \\ \cline{2-4} 
\multicolumn{1}{|c|}{} & \multicolumn{1}{l|}{Small} & \multicolumn{1}{l|}{Medium} & Large \\ \hline
500.perl\_r & \multicolumn{1}{l|}{16 KB} & \multicolumn{1}{l|}{248 KB} & 600 KB \\ \hline
502.gcc\_r & \multicolumn{1}{l|}{220 KB} & \multicolumn{1}{l|}{2.3 MB} & 4.5 MB \\ \hline
505.mcf\_r & \multicolumn{1}{l|}{4 KB} & \multicolumn{1}{l|}{4 KB} & 4 KB \\ \hline
508.namd\_r & \multicolumn{1}{l|}{4 KB} & \multicolumn{1}{l|}{4 KB} & 4 KB \\ \hline
519.lbm\_r & \multicolumn{1}{l|}{4 KB} & \multicolumn{1}{l|}{4 KB} & 4 KB \\ \hline
520.omnet\_r & \multicolumn{1}{l|}{21 MB} & \multicolumn{1}{l|}{206 MB} & 1.8 GB \\ \hline
523.xalanc\_r & \multicolumn{1}{l|}{116 KB} & \multicolumn{1}{l|}{25 MB} & 4.9 GB \\ \hline
525.x264\_r & \multicolumn{1}{l|}{8 KB} & \multicolumn{1}{l|}{8 KB} & 20 KB \\ \hline
526.blender\_r & \multicolumn{1}{l|}{260 KB} & \multicolumn{1}{l|}{312 KB} & 320 KB \\ \hline
531.deep\_r & \multicolumn{1}{l|}{8 KB} & \multicolumn{1}{l|}{16 KB} & 16 KB \\ \hline
538.imag\_r & \multicolumn{1}{l|}{8 KB} & \multicolumn{1}{l|}{12 KB} & 16 KB \\ \hline
541.leela\_r & \multicolumn{1}{l|}{12 KB} & \multicolumn{1}{l|}{12 KB} & 16 KB \\ \hline
544.nab\_r & \multicolumn{1}{l|}{4 KB} & \multicolumn{1}{l|}{4 KB} & 4 KB \\ \hline
557.xz\_r & \multicolumn{1}{l|}{4 KB} & \multicolumn{1}{l|}{4 KB} & 4 KB \\ \hline
\end{tabular}
\caption{WPP based function profile Disk Size after \textit{Morpheus}'s Smart-Loop Instrumentation. In majority of the benchmarks, the profile sizes are reduced to KBs from GBs.}
\label{res:tab1}
    \end{minipage}
    &
    \begin{minipage}{.51\textwidth}
    \centering\includegraphics[width=1.0\textwidth]{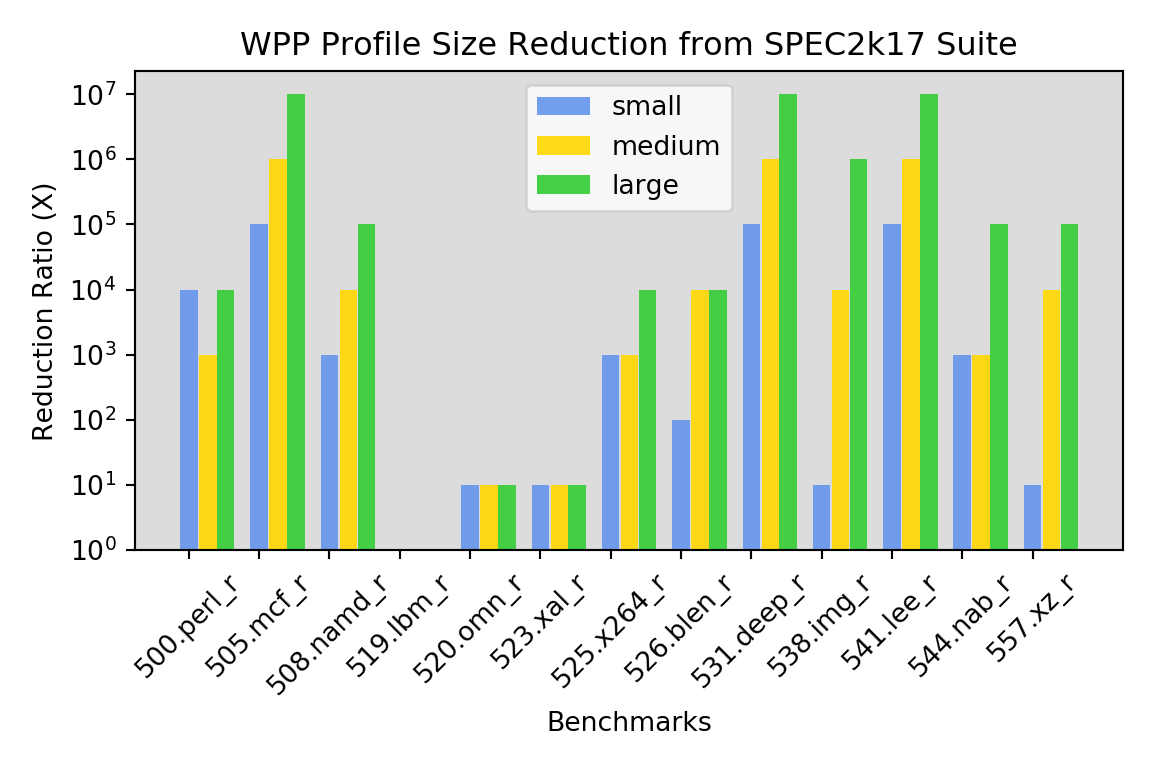}
    %\vspace{-15pt}
\caption{Ratio of Reduction in Profile Size achieved by Smart-loop Instrumentation over Normal WPP based function profiles. The trend shows that \textit{Morpheus} achieves higher degree of compression for largest input sizes, which is in line with the higher redundancies in larger profiles.} 
    \label{fig:res2}
\end{minipage}
\end{tabular}
\end{figure*}

$\star$ \textbf{Input Consistent-Compression \& Augmentation}: Fig. \ref{fig:res3a} depicts the results of the input-consistent compression scheme, where each compressible region is statically assigned an encoding. In majority of the benchmarks, the total regions are below 20. The largest number of compressible regions identified in a benchmark was 498 (\textit{526.blender\_r}). The average function call length of these regions was nearly 3 function calls per compressible regions across all benchmarks (Fig. \ref{fig:res3b}). 
\begin{figure*}[ht]
\begin{tabular}{p{0.5\textwidth}p{0.5\textwidth}}
    \begin{minipage}{.49\textwidth}
    \centering\includegraphics[width=1.0\textwidth]{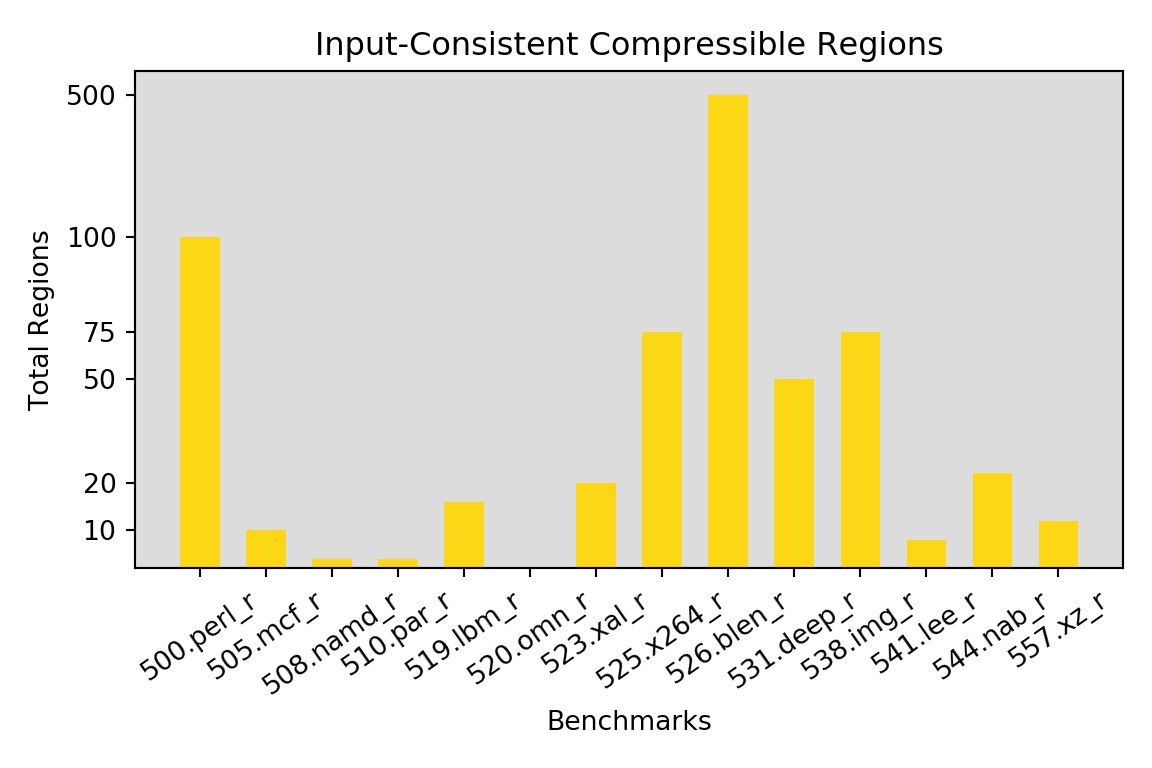}
\caption{Input-consistent compressible regions identified statically from SPEC benchmarks. } 
\label{fig:res3a}
    \end{minipage}
    &
    \begin{minipage}{.49\textwidth}
    \centering\includegraphics[width=1.0\textwidth]{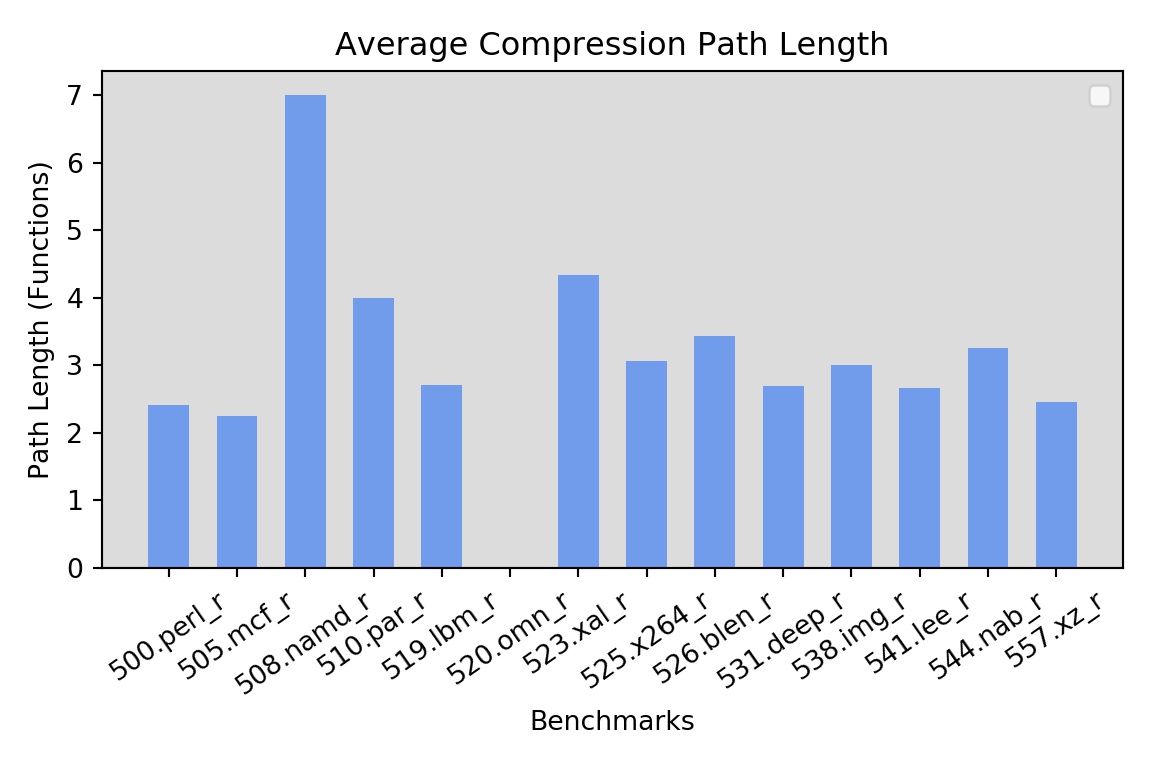}
\caption{Average length of input-consistent compressible paths for each benchmark. } 
    \label{fig:res3b}
\end{minipage}
\end{tabular}
\end{figure*}

\begin{figure*}[htbp]
\begin{tabular}{p{0.5\textwidth}p{0.5\textwidth}}
    \begin{minipage}{.49\textwidth}
    \centering\includegraphics[width=1.0\textwidth]{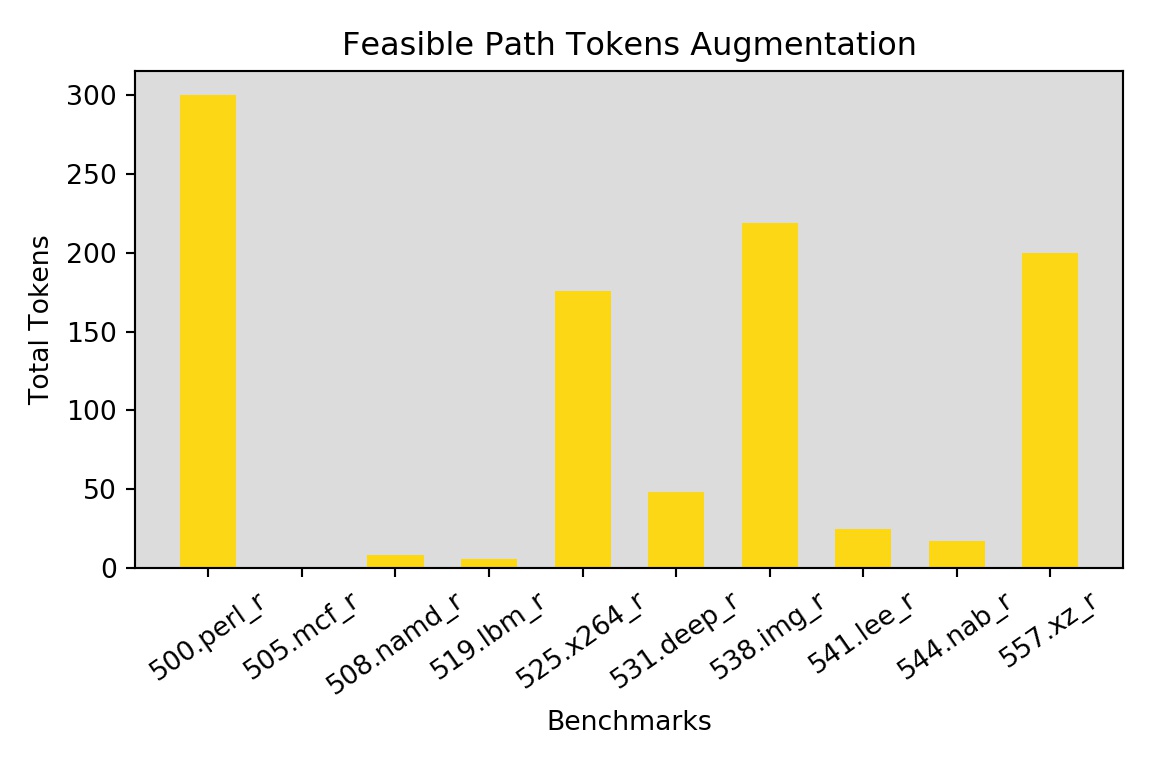}
\caption{Total tokens augmented in WPP based function profile for capturing program context.} 
\label{fig:res3c}
    \end{minipage}
    &
    \begin{minipage}{.49\textwidth}
    \centering\includegraphics[width=1.0\textwidth]{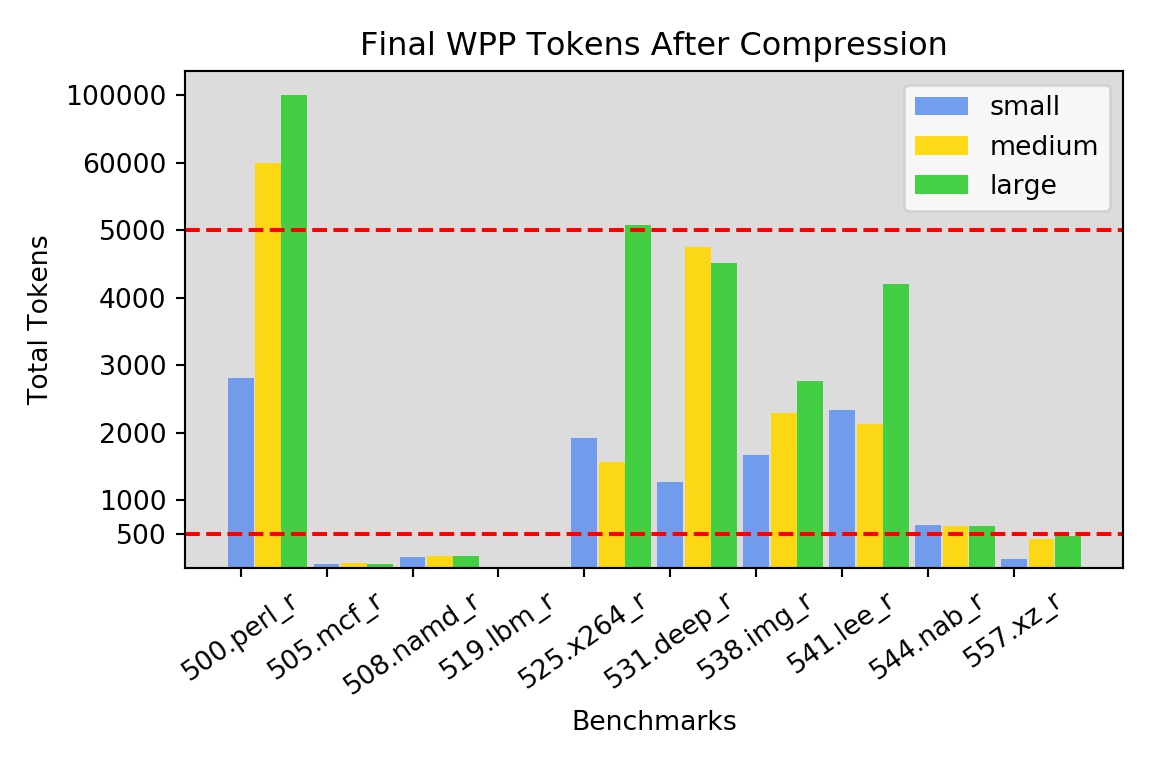}
\caption{Final Token Count in WPP based function profile after input-consistent compression.} 
    \label{fig:res3d}
\end{minipage}
\end{tabular}
\end{figure*}

$\star$ \textbf{Determining Computationally Frequent Functions}: Fig. \ref{fig:res4a} shows the computationally \textit{hot} functions predicted by \textit{Morpheus} on unseen program inputs. As observed, the in all cases, \textit{Morpheus} is able to predict close to 99\% of the total function calls in the program, by only predicting a fraction of total functions present in the application. This is explained by the fact that the RNN-based model can effective capture the sequence of functions captured in loop-intensive regions, and is adapts to the unseen inputs by activating the tokens corresponding to the unseen augmented paths. 

Fig. \ref{fig:res4b} shows the training time required by \textit{Morpheus} to learn the generalized WPP based function profiles. As observed, \textit{Morpheus} can complete the training process in minutes without deploying any specialized hardware. The fast training is a direct consequence of \textit{Morpheus} pre-processing the WPP based function profile into input-consistent compressible regions, which reduces the total tokens. Furthermore, the redundancy removal by the smart loop instrumentation scheme reduces bias, thereby preventing the model results from being skewed towards certain functions.   

\begin{figure*}[ht]
\begin{tabular}{p{0.5\textwidth}p{0.5\textwidth}}
    \begin{minipage}{.49\textwidth}
    \centering\includegraphics[width=1.0\textwidth]{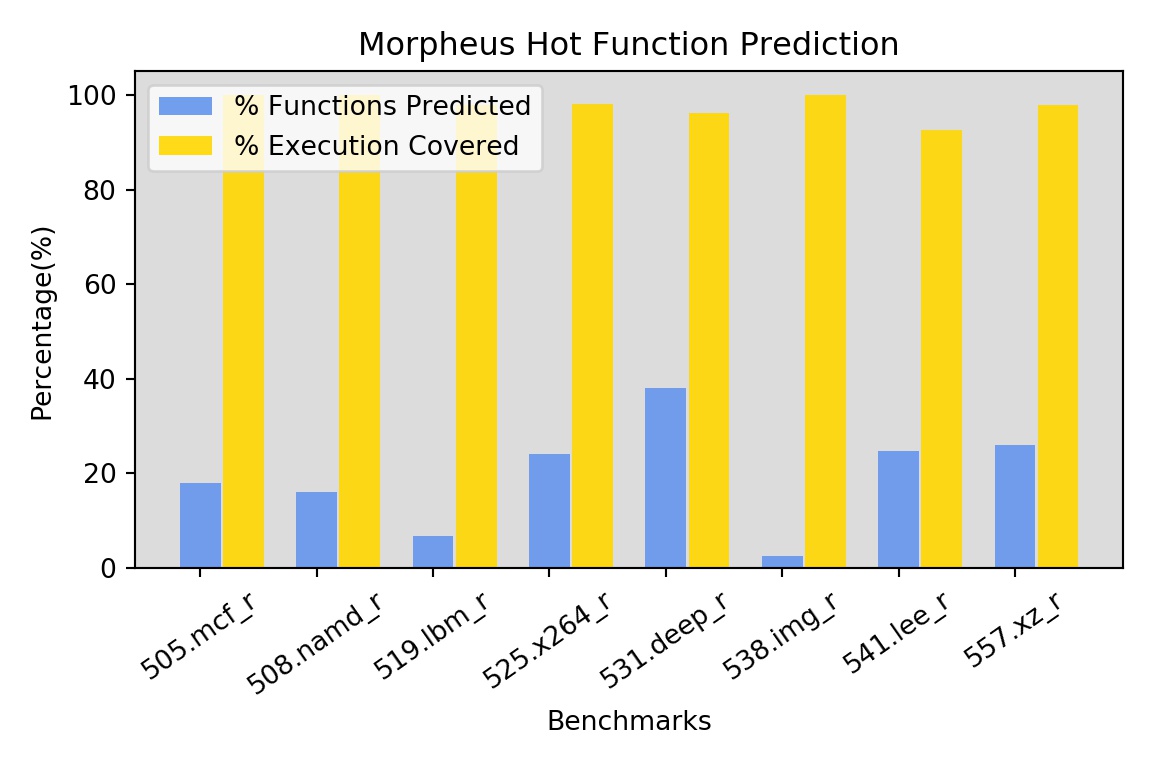}
\caption{\textit{Morpheus} predicting of computationally \textit{most frequently executed} functions, as fraction of total functions and execution covered. The trend shows that \textit{Morpheus} is able to predict almost 99\% of program's execution with only a fraction of functions.} 
\label{fig:res4a}
    \end{minipage}
    &
    \begin{minipage}{.49\textwidth}
    \centering\includegraphics[width=1.0\textwidth]{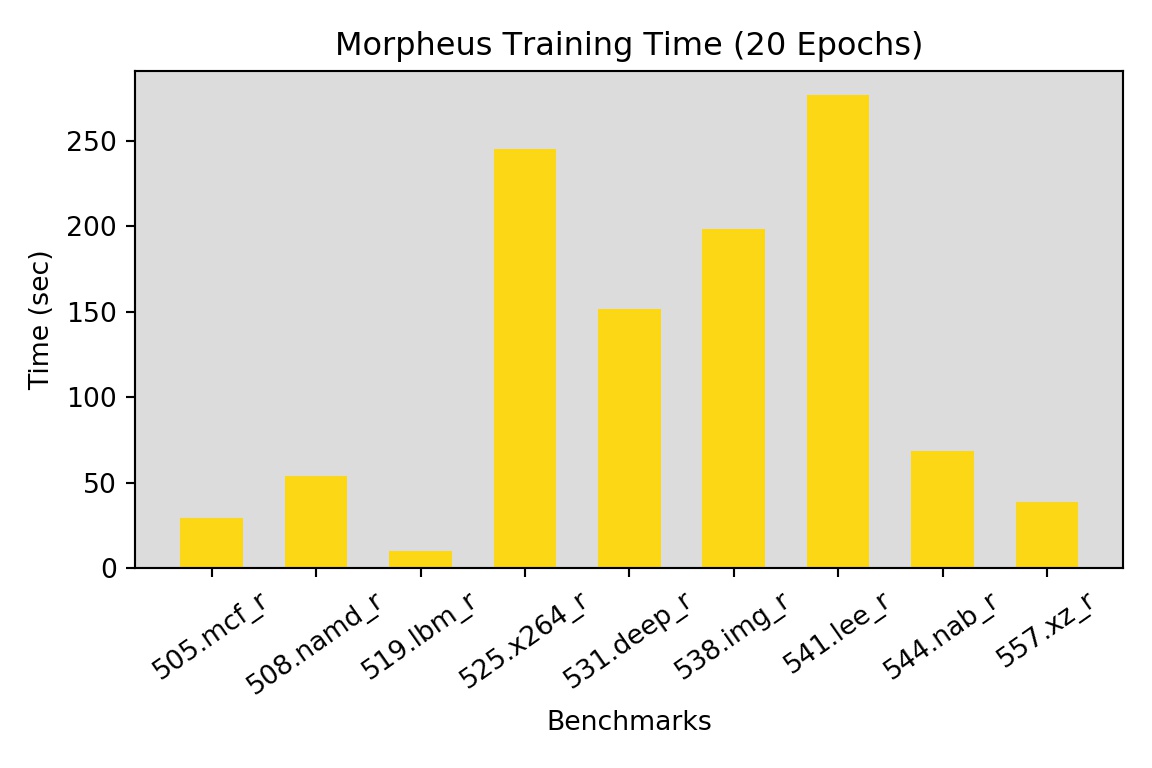}
\caption{\textit{Morpheus}'s training time for generative models. The RNNs were trained for 20 epochs in each case, with varying batch sizes ranging from 4 to 64. \textit{Morpheus} is able to complete training within minutes, without any specialized accelerators.} 
    \label{fig:res4b}
\end{minipage}
\end{tabular}
\vspace{-1em}
\end{figure*}

$\star$ \textbf{Scaling of Morpheus with different application sizes}: Fig. \ref{fig:scale} shows Morpheus’s compressed profile generation time across applications of increasing size. As observed, the compression time grows log-linearly with code size. This is especially highlighted for large codebases such as 526.blender\_r, where the major bottleneck is the online compaction scheme of detecting longest repeating call chains. Across all benchmarks, larger input datasets consistently take longer, confirming that Morpheus’s cost scales with both program and input size.

\subsection{Dynamis Vs Morpheus: performance tradeoffs and considerations}
\label{res:vs}

In this subsection, we explore the utility and performance tradeoffs of both \textit{Dynamis} and \textit{Morpheus}; compared against one another. The primary necessity of a single profile-driven approach such as \textit{Morpheus} arises from the fact that with the default set of prompts and model weights, the LLM is unable to provide any sequential ordering of the function calls during runtime. For instance, querying GPT 4.0 to reason about the dynamic call for the smallest input in \textit{519.lbm\_r} generates the following trace (missing tokens represented in {\color{red}red}):
\begin{align*}
 \text{LLM-generated Trace}: 0, 1, 2, 4, 5, 6, 8, \{10,11,12\}\times 20 ~~~~~\\
 \text{Actual Trace}: 0, 1, 2, 4, {\color{red}4}, 5, {\color{red}5}, 6, {\color{red}6, 7, 7, 9,} \{10,11,12\} \times 20, {\color{red} 9, 3, 3}
\end{align*}

\begin{wrapfigure}{r}{0.45\textwidth}
%\vspace{pt}
  \begin{center}
    \includegraphics[width=0.45\columnwidth]{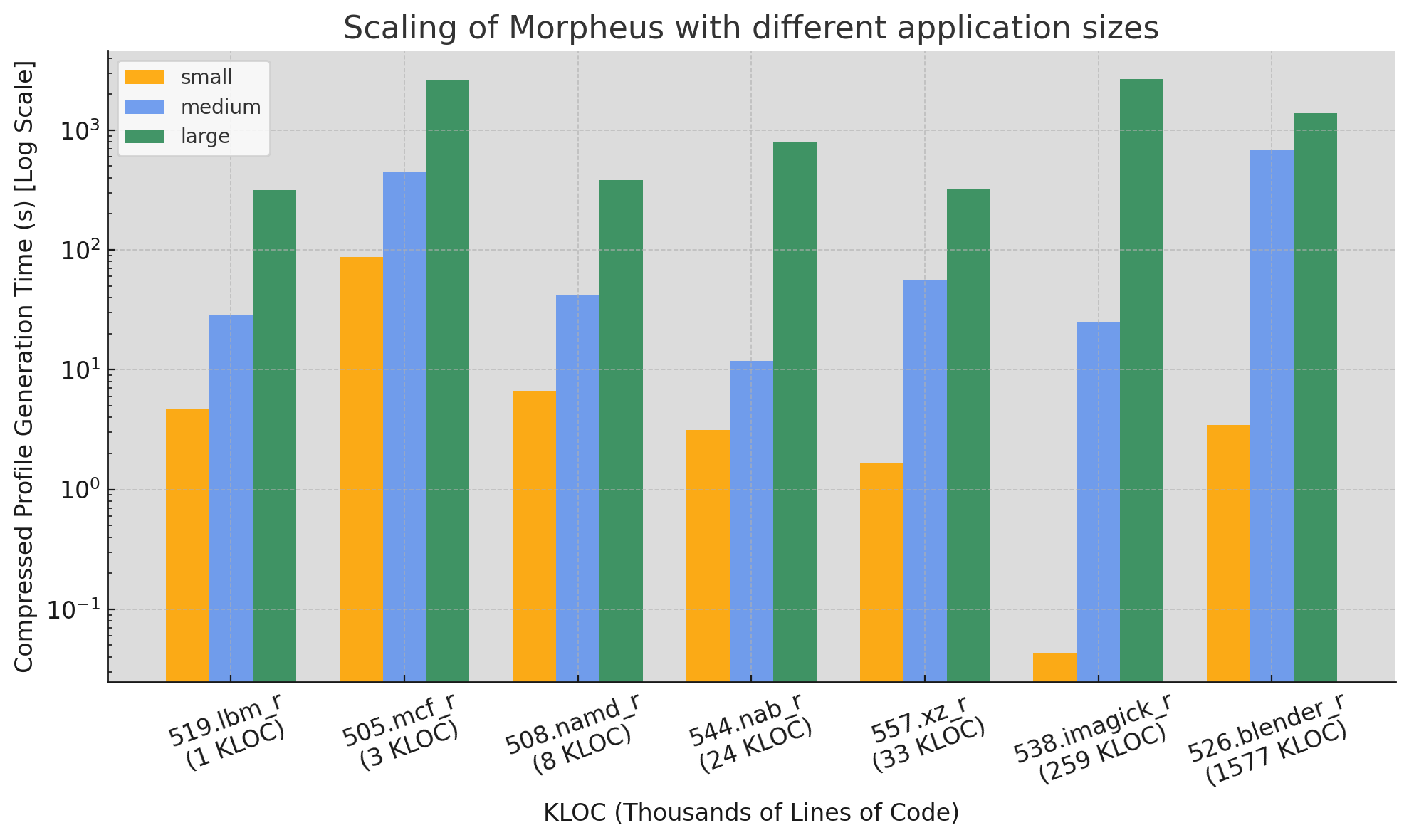}
\caption{\small \textit{Morpheus}' Log-linear scaling of compression time with larger inputs and bigger codebases}
\label{fig:scale}
  \end{center}
  \vspace{-15pt}
\end{wrapfigure}

For reference, the entire control flow is represented in Fig. \ref{fig:0}. Thus, in order to obtain a clearer picture of an application's dynamic sequence of function calls, we employ profile generalization. As described earlier in \S\ref{res:morph}, \textit{Morpheus} is able to achieve upto 99\% of the program's execution. However, achieving \textit{application profile generation} and successfully leveraging it for profile-driven tasks still requires the users to generate a single application profile by executing the program on a representative input. Furthermore, there is an implicit assumption about the user's necessary expertise to determine the convenient input set on which the application profile needs to be collected. For specific long-running applications, the task of program execution and data collection could potentially result in a bottleneck. For instance, Fig. \ref{fig:time1} shows the distribution of time required to collect WPP traces. As observed in the figure, profiling certain execution instances can even take up to 3 hours (\textit{541.leela\_r}), which can complicate the process of profile generalization. Thus, we still need to employ profile-less approach such as \textit{Dynamis}. Another noteworthy point in the previous example is that inspite of the obvious inaccuracies in the program trace, LLM was not only able to predict the hot functions (10, 11, 12), but also the frequencies accurately. We anticipate that with better prompting and sophisticated fine-tuning, LLMs can be used to bridge the gap in function call chain prediction as well. 

\section{Related works}
\label{sec:related}

%Recent research into code optimization \cite{gao2024searchbasedllmscodeoptimization, rosas2024aioptimizecodecomparative, ma2024llamocoinstructiontuninglarge} has evolved from traditional rule-based methods, which often suffer from limited coverage and labor intensity, to approaches leveraging large language models (LLMs) for automated refactoring \cite{pomian2024furtherllmsidestatic, cordeiro2024empiricalstudycoderefactoring}.
%LLM-based methods, such as those employing sequence generation, show promise but often struggle with complex tasks and require iterative refinement frameworks for effective optimization.
%These frameworks integrate LLMs with evolutionary search, adaptive pattern retrieval, and specialized prompting \cite{liu2023promptingframeworkslargelanguage, wu2024evolutionarycomputationeralarge}.
%Comparisons between LLMs (e.g., GPT-4.0 and CodeLlama-70B) and traditional optimizing compilers (CETUS, PLUTO, ROSE) \cite{rosas2024aioptimizecodecomparative} reveal that while LLMs can achieve greater speedups, they may generate incorrect code on large inputs, necessitating verification.
%Some of the recent works using instruction-tuning framework \cite{ma2024llamocoinstructiontuninglarge} address LLM limitations by fine-tuning models with comprehensive instruction sets, yielding superior optimization performance across synthetic and real-world tasks.

Recent research in code optimization has transitioned from rule-based compiler heuristics to large language model (LLM)–driven code transformation and refactoring \cite{gao2024searchbasedllmscodeoptimization, rosas2024aioptimizecodecomparative, ma2024llamocoinstructiontuninglarge, pomian2024furtherllmsidestatic, cordeiro2024empiricalstudycoderefactoring}. LLM frameworks combine sequence generation with iterative refinement using evolutionary search, adaptive retrieval, and specialized prompting \cite{liu2023promptingframeworkslargelanguage, wu2024evolutionarycomputationeralarge}. Comparative studies show that LLMs can rival or surpass traditional compilers such as CETUS and PLUTO in optimization quality but often struggle with correctness and scalability on large real-world programs. Instruction-tuned models \cite{ma2024llamocoinstructiontuninglarge} further address these issues through optimization-aware fine-tuning, improving reliability and consistency across diverse workloads.

In parallel, \textit{machine learning–based static profiling} has emerged as an alternative to dynamic PGO, eliminating the need for costly runtime instrumentation and representative input collection. VESPA \cite{vespa2021moreira} pioneered static frequency prediction for binary optimization in Meta’s BOLT \cite{panchenko2018boltpracticalbinaryoptimizer}, achieving up to 6\% speedup over Clang –O3 by estimating basic-block frequencies directly from static features. \textit{GraaSP} \cite{CUGUROVIC2024112058} extended this concept to compiler integration within \textit{GraalVM} \cite{graalvm2020}, embedding ML inference during compilation to remove the two-phase PGO cycle and delivering a 7.5\% runtime gain. GraalNN \cite{milikic2025graalnn} advanced the field further by employing Graph Neural Networks to model call-graph context, refining predictions across parsing and inlining phases for greater accuracy. Building on these insights, our work explores input-specific static profiling without execution, using LLMs to infer hot regions based on algorithmic structure, input size, and termination behavior. Unlike prior systems targeting speed in Java or binary infrastructures, our framework focuses on code optimization in C/C++ workloads, evaluated against PGO baselines on SPEC CPU 2017. Though not directly comparable, our results demonstrate that LLMs can approximate—and often surpass—traditional profiling accuracy while offering fine-grained, input-aware guidance.

On the other hand, the concept of capturing comprehensive dynamic program behavior through Whole Program Paths (WPP) has been extensively explored in the compiler literature \cite{larus1999whole, zhang2001timestamped, burtscher2003compressing, zhang2002path, law2003whole, milenkovic2003stream, zhang2004whole, lin2007design}. WPPs provide a single, compact description of a program's entire dynamic control flow, including interprocedural calls and loop iteration, which is essential for finding "hot subpaths" for optimization. While techniques for compressing WPP profiles are well-developed, they are inherently limited to analyzing a single profile instance generated from one specific input data set. Our approach, in contrast, attempts to generalize across multiple profiles by using a novel, lossless compression scheme that is consistent across different inputs, enabling the prediction of program behavior for unseen execution scenarios.
%\textit{Whole Program Profile Paths} has been explored in-depth in the compiler literature \cite{larus1999whole, zhang2001timestamped, burtscher2003compressing, zhang2002path, law2003whole, milenkovic2003stream, zhang2004whole, lin2007design}. These works include analyzing one single profile at a time, which involves querying a compacted representation of the profiles to answer questions pertaining to the application. Our approach, on the other hand, attempts to reason about multiple profiles with the objective of generalizing across different execution scenarios.   

\section{Conclusion}
\label{sec:conc}

This paper introduced \textit{Phaedrus}, a novel framework that advances profile-driven software optimization through predictive modeling of dynamic application behavior.  Phaedrus offers a unique hybrid approach that combines \textit{Application Profile Generalization} and \textit{Application Behavior Synthesis}. Our framework uses generative deep learning models and large language models (LLMs) integrated with static compiler analysis to accurately predict function calls and control flow behavior, enabling adaptive optimizations without the need for exhaustive profiling across diverse inputs. We show that LLMs provide a good prediction capability of hot functions in an input specific manner that yield performance that is comparable to and slightly better than traditional PGO and code size reduction comparable to traditional PGO. RNN based framework improves performance over traditional PGO by 2.8\%  although an initial investment of doing profiling and model generation is needed. Due to its near perfect accuracy, a key use of Morpheus could be to perform offline function call tracing for doing audits whereas LLMs yield a zero profile approach resulting in a simplified software process.
\section{Data Availability Statement}
\label{sec:datavail}

This work involves extensive experimentation with large language models (Claude 4 Sonnet, GPT-4, GPT-5) and other machine learning models, combined with static analysis and profiling techniques. We rely on established benchmark suites such as SPEC CPU 2017 and GAPS, and use a variety of profiling tools to establish ground truths for evaluating model predictions. We plan to submit our artifacts for the Artifact Evaluation process. While the artifacts are still being finalized, we intend to make available: (i) the source code for the benchmarks used; (ii) static analysis outputs derived from these codebases; (iii) model predictions; and (iv) evaluation scripts. These materials will enable reviewers and researchers to reproduce and build upon our results. Preliminary versions of the artifacts will be anonymized and linked as part of the artifact submission.

%%
%% The acknowledgments section is defined using the "acks" environment
%% (and NOT an unnumbered section). This ensures the proper
%% identification of the section in the article metadata, and the
%% consistent spelling of the heading.

%\begin{acks}
%To Robert, for the bagels and explaining CMYK and color spaces.
%\end{acks}

%%
%% The next two lines define the bibliography style to be used, and
%% the bibliography file.
\bibliographystyle{ACM-Reference-Format}
\bibliography{sample-base}

%%
%% If your work has an appendix, this is the place to put it.
\appendix

\end{document}